\documentclass{article}%

\usepackage[title]{appendix}
\usepackage{setspace}

\usepackage{xr-hyper}
\usepackage{hyperref}
\usepackage{amsmath}
\usepackage{amsfonts}
\usepackage{amssymb}
\usepackage{amsthm}
\usepackage{enumitem}
\usepackage{cases}
\usepackage{cleveref}
\usepackage{tensor}
\usepackage{booktabs}
\usepackage{bm}
\usepackage{comment}
\usepackage[normalem]{ulem}

\newcommand{\R}{\mathbb{R}}
\newcommand{\E}{\mathbb{E}}
\newcommand{\N}{\mathbb{N}}
\newcommand{\G}{\mathbb{G}}
\newcommand{\Z}{\mathbb{Z}}
\newtheorem{theorem}{Theorem}[section]
\newtheorem{corollary}[theorem]{Corollary}

\newtheorem{lemma}[theorem]{Lemma}

\theoremstyle{definition}
\newtheorem{definition}[theorem]{Definition}
\newtheorem{example}[theorem]{Example}
\DeclareMathOperator{\Image}{Im}

\DeclareMathOperator{\Cov}{Cov}
\DeclareMathOperator{\dom}{dom}
\DeclareMathOperator{\Aut}{Aut}
\DeclareMathOperator{\sign}{sign}
\DeclareMathOperator{\image}{Im}
\DeclareMathOperator{\ob}{ob}

\DeclareMathOperator{\Set}{Set}

\newcommand{\ed}{\mathrm d}
\newcommand{\id}{\mathrm{id}}
\newcommand{\AMPR}{\mathrm{AMPR}}
\newcommand{\FinM}{{\mathrm{FinM}}}
\newcommand{\VecM}{{\mathrm{VecM}}}
\newcommand{\DualM}{{\mathrm{DualM}}}
\newcommand{\Dual}{{\mathrm{Dual}}}

\newcommand{\Measures}{{\mathrm{Measures}}}
\renewcommand{\Vec}{{\mathrm{Vec}}}
\newcommand{\Group}{{\mathrm{Group}}}
\newcommand{\Top}{{\mathrm{Top}}}

\newcommand{\Prob}{{\mathrm{Prob}}}

\newcommand{\Fin}{{\mathrm{Fin}}}

\renewcommand{\P}{\mathbb P}
\newcommand{\Q}{\mathbb Q}

\begin{document}

\title{Classifying Financial Markets up to Isomorphism}

\author{
John Armstrong}
\date{}

\maketitle

\begin{abstract}
Two markets should be considered isomorphic
if they are financially indistinguishable.
We define a notion of isomorphism for financial markets
in both discrete and continuous time. We then seek to identify
the distinct isomorphism classes, that is to classify markets.

We classify complete one-period markets. We define an invariant of 
continuous time complete markets which we call the absolute
market price of risk. This invariant plays a role analogous
to the curvature in Riemannian geometry. We classify markets when
the absolute market price of risk is deterministic.

We show that, in general, markets with
non-trivial automorphism groups admit mutual fund theorems. We prove
a number of such theorems.
\end{abstract}

\section*{Introduction}

Two financial markets should be considered equivalent if there is a bijective correspondence between the investment strategies in each market which preserves both the costs and the payoff distributions of these strategies. This intuition allows us to define a formal notion of isomorphism for financial markets.  In the language of category theory 
\cite{eilenbergmaclane}, we shall define the category of financial markets.

We shall then demonstrate that one can prove financially interesting classification theorems. We classify Gaussian markets and complete one-period markets. We also prove a partial classification theorem for complete continuous-time markets with a fixed risk-free rate which we will now describe.

The minimum number of assets required to replicate an arbitrary contingent claim gives one basic invariant of such markets, the dimension. The next useful 
invariant we identify is the length of the market-price-of-risk vector, which we call the {\em absolute market price of risk}. While it
is easy to define other invariants, this has the advantage of being a {\em local invariant}, by which we mean that it can be calculated from 
the coefficients of an SDE defining the asset price dynamics by simple algebra and differentiation. The absolute market price of 
risk gives a basic invariant of markets up to isomorphism. In this sense it is analogous to the Riemannian  curvature, which
gives a basic invariant of Riemannian manifolds up to isometry.

We classify continuous-time complete markets whose absolute market price of risk is deterministic. Markets with constant absolute market price of risk are determined up to isomorphism by just their dimension, the risk-free rate and the absolute market price of risk and are isomorphic to Black--Scholes--Merton markets.

\medskip

Our classification theorems have a number of interesting financial applications.

Firstly, one can often use a classification theorem to illuminate a mathematical proof using without-loss-of-generality arguments.
We will see that one can specify an $n$-dimensional Black--Scholes--Merton market up to isomorphism
using only the parameters of dimension, risk-free rate
and absolute market price of risk. This allows one to prove financial results for these markets by considering only markets with
particularly simple forms. Our classification of one-period complete markets
admits similar applications.

Secondly, we establish a connection between the automorphisms of a market and mutual-fund theorems. We prove
that investment strategies solving invariant convex optimization problems in a market can be assumed to be invariant under automorphisms. For
markets with large symmetry groups such as Black--Scholes--Merton markets, this imposes strong limitations on the form of optimal investment
strategies, giving a significant generalization of the classical mutual-fund theorems.

Thirdly, we will see that a surprisingly large number of markets are isomorphic to a Black--Scholes--Merton market, and so
financial results proved for such markets can be applied more widely than one might expect. In particular given any diffusion model,
one can obtain a related Black--Scholes--Merton market by making an appropriate choice of drift. This is 
significant since the drift is difficult to estimate from statistical evidence and its functional form is usually chosen for parsimony.
This result implies that one can
find, for example, stochastic-volatility models which are isomorphic to Black--Scholes--Merton markets. 

\medskip

The effect of transformations on a market has been considered by many authors. If one considers the asset prices as stochastic trajectories in $\R^n$ one can ask how the dynamics change under diffeomorphisms of $\R^n$. Stochastic differential equations (SDEs) on manifolds have been studied extensively and this has lead to a variety of geometric formulations \cite{armstrongBrigoJets, belopolskaya2012stochastic,elworthy, emery, gliklikh2010global, itoManifold, hsu}. The diffusion term of a non-degenerate SDE on a manifold can be interpreted as defining a Riemannian metric and this yields a connection between Riemannian curvature and SDEs. The geometric theory of SDEs on manifolds has been successfully applied to finance in, for example, \cite{henryLabordere}. However,
the maps induced by diffeomorphisms of $\R^n$ are hard to interpret financially
since financially important properties, such as whether a process is a martingale, are not preserved by diffeomorphisms.

In this paper, the transformations we consider are those that do preserve financially important properties. They are given by maps between the underlying probability spaces defining the markets rather than on the space $\R^n$. The objects in our categories are given by filtered probability spaces equipped with cost functionals. This probabilistic definition of a market is influenced by the work of Pennanen \cite{pennanen, pennanenDuality}. The morphisms we define
are built upon the theory of probability-space homomorphisms developed by Rokhlin in \cite{rokhlin} (who extended the work of von Neumann in \cite{vonNeumann}).

Rokhlin's classifications
of standard probability spaces and their homomorphisms are key ingredients in our classification of one-period complete markets. While Rokhlin's results are all we need for this paper, we note that the category theory of probability has been developed by other authors, for example the theory of stochastic processes 
is explored from a categorical viewpoint in \cite{giryCategorical}.

Category theory was applied to financial markets in
\cite{armstrongMarkowitz}, which classifies Markowitz markets 
and relates the classical mutual-fund theorems to the symmetries of the market.
The formulation of this earlier paper is purely algebraic. The 
probabilistic formulation we will develop is more fundamental and more general.

\medskip

The structure of the paper is as follows.

In Section \ref{sec:markowitz} we define the category of discrete-time markets. We prove a general mutual-fund theorem for
markets with automorphisms. We prove the equivalence (or more precisely the duality) between the formulation of categories given in this paper
and the algebraic approach of \cite{armstrongMarkowitz}. We illustrate our mutual-fund theorem with the example of Gaussian markets.

In Section \ref{sec:completeOnePeriod} we classify complete one-period markets. We first give a simplified classification by assuming that one can additionally invest in a ``casino'', which is a complete market where the $\P$ and $\Q$ measures coincide. This form of the classification is sufficient for most applications.
We give a full classification for markets without using the casino. Our general mutual-fund theorem only applies to convex optimization problems, but
for complete one-period markets we are also able to prove a mutual-fund theorem for problems where all agents have monotonic preferences.
This generalization is useful for solving problems involving investors with S-shaped utility, as motivated by the theory of Kahneman and Tversky \cite{kahnemanTversky}.

In Section \ref{sec:ctstime} we extend out category to multi-period and continuous-time markets. We classify complete continuous-time markets
with constant absolute market price of risk.

Since the primary novelty of this paper is our definitions of financial categories,
the resulting classification results and their financial applications, the proofs of our results have been placed in an online appendix.
If the reader is unfamiliar with category theory, a short review of the basic terminology
we require can be found in a second appendix. The appendices may be found on the journal website or in the arXiv version of this manuscript \cite{classificationArxiv}.

\section{Finite-dimensional linear markets}
\label{sec:markowitz}

In this section we give a coordinate-free definition of a one-period financial market and relate this to the elementary, coordinate-based approach of defining a market in $n$-assets using a probability distribution
on $\R^n$. We will illustrate
with the example of the Markowitz model. We will use this to demonstrate
the relationship between invariant investment strategies
and mutual-fund theorems.

We begin by recalling a number of definitions due to Rokhlin \cite{rokhlin} for
morphisms between probability spaces.

\begin{definition}
	Let $(\Omega_1, {\cal F}_1, \P_1)$ and $(\Omega_2, {\cal F}_2, \P_2)$  be two	
	probability spaces. A map $\phi:\Omega_1 \to \Omega_2$ is called a {\em homomorphism}
	if $\phi$ is measurable and if $\P_1(\phi^{-1}U)=\P_2(U)$ for all $U \in {\cal F}_2$.
	A homomorphism $\phi$ is called an {\em isomorphism} if it is bijective and its inverse is a homomorphism. We call $\phi$ a {\em mod 0 isomorphism} if there are subspaces $\Omega_1^\prime \subseteq \Omega_1$ and $\Omega_2^\prime \subseteq \Omega_2$ both of full measure such that $\phi$ restricted to $\Omega_1$ is an isomorphism to $\Omega_2$. 
\end{definition}

\noindent From the point of view of probability theory, two probability spaces should
be considered as equivalent if they are mod 0 isomorphic.
We define the category $\Prob$ to have objects given by probability spaces and morphisms given
by almost-sure equivalence classes of homomorphisms. Rohklin's definition of a mod 0 isomorphism does not
coincide exactly with the set of isomorphisms in $\Prob$. The next lemma explains how the two notions are related.

\begin{lemma}
	A measurable function is a mod 0 isomorphism if and only if its almost-sure equivalence class is a $\Prob$ isomorphism.
	\label{lemma:mod0Lemma}
\end{lemma}

An important functor is the contravariant functor $L^0$ which maps the category $\Prob$ to the category $\Vec$ of vector spaces. $L^0$ acts on the objects of $\Prob$ by mapping a probability space to its vector space of almost-sure equivalence classes of measurable functions. Given a $\Prob$ morphism $f:\Omega_1 \to \Omega_2$ and $X \in L^0(\Omega_2)$ we define  a linear transformation $L^0(f):L^0(\Omega_2)\to L^0(\Omega_1)$ by $L^0(f)(X)=X \circ f$.

\begin{definition}
	A {\em one-period financial market} $((\Omega, {\cal F},\P),c)$ consists of: a probability space $(\Omega, {\cal F},\P)$; a function $c:L^0(\Omega; \R) \to \R \cup \{\pm \infty\}$. 
	We call $c^{-1}(\R \cup \{-\infty\})$ the {\em domain} of $c$, denoted $\dom c$.
\end{definition}

We interpret a real random variable $X$ on $\Omega$ as an investment strategy with
payoff $X(\omega)$ in scenario $\omega \in \Omega$. $c(X)$ denotes the up front cost of strategy
$X$ and is equal to $\infty$ if one cannot pursue a strategy. A strategy with $c(X)=-\infty$ results in liabilities so bad that the market is willing to pay arbitrarily large incentives to encourage someone to take these liabilities on. A typical investment strategy is the purchase of an asset or of a portfolio of assets which are then sold at a final time $T$. In this case $c(X)$ would be the cost of purchasing the asset. However, one can also model a commitment to pursue a continuous-time trading strategy as yielding a single payoff at the final time $T$
and our definition of a market is flexible enough to include such strategies.

This definition is deliberately minimal. To obtain interesting markets one would typically want to impose additional conditions, such as that the market should be arbitrage free. This condition can be expressed as: for random variables $X$, if $X \geq 0$ and $X \neq 0$ then $c(X)>0$.

In this section we will be interested only on one-period markets so we will 
refer to them simply as markets.
\begin{definition}
	\label{def:morphism}
	A {\em morphism} of markets $M_1=((\Omega_1, {\cal F}_1,\P_1),c_1)$ and $M_2=((\Omega_2, {\cal F}_2,\P_2),c_2)$ is a
	$\Prob$ morphism $\phi:\Omega_1 \to \Omega_2$ satisfying
	$c_2(X) \geq c_1(X\circ \phi)$ for all $X \in L^0(\Omega_2; \R)$.
\end{definition}

Financially, a market morphism $\psi:M_1 \to M_2$ represents an inclusion of the market $M_2$ in $M_1$: given
an investment strategy represented by the random variable $X$ in $M_2$, we have the investment strategy $X \circ \psi$ in $M_1$ which  has identical payoff distribution but which has lower up-front cost. So if one can afford to pursue the strategy $X$, one can also afford to pursue $X \circ \psi$.
The contravariance between $\psi$ and the financial notion of inclusion stems from the contravariance of the functor $L^0$.

Our primary interest in this paper is in market isomorphisms. We may describe them as follows.
\begin{lemma}
	\label{lemma:isomorphism}
	An isomorphism of markets $((\Omega_1, {\cal F}_1,\P_1),c_1)$ and $((\Omega_2, {\cal F}_2,\P_2),c_2)$ is the
	almost-sure equivalence class of a mod 0 isomorphism $\phi:\Omega_1 \to \Omega_2$ satisfying
	$c_2(X)=c_1(X\circ \phi)$ for all $X \in L^0(\Omega_2; \R)$.
\end{lemma}

\medskip

In finance, optimal investment problems are often convex optimization problems (see for example \cite{pennanen}). A convex optimization problem is a problem requiring finding the set of minimizers of a convex objective function over a convex domain. For example, a risk-averse agent will have a concave utility function, and so the objective in expected utility maximization problems can be expressed as the minimization of their convex expected disutility function. Cost constraints are typically linear, and hence define a convex domain. Additional constraints imposed by a risk manager will further restrict the domain, but if one uses expected-shortfall constraints, or any other coherent, or simply convex, risk measure (see \cite{follmerSchied}), this too will yield a convex domain.

The solution set of a convex optimization problem is itself a convex set.
We also expect that if the solution set is financially meaningful, it will be invariant under the automorphism group of the market. Our next result will show that one may then find an element of the solution set which is itself invariant.

To state our result, let us define the necessary terminology. A measurable group $G$ has a {\em left-invariant} probability measure, $\G$ if for all measurable sets $A \subseteq G$ and elements $h \in G$ we have $\G(A)=\G(h A )$. A {\em representation} of such a group on a Banach space $V$ is a group homorphism $\rho:G \to \Aut V$, where $\Aut V$ is the group of linear isometries of $V$. We think of $\rho$ as defining an action of $G$ on $V$ on the left, given by $g v$=$\rho(g) v$.

\begin{theorem}
	Let $G$ be a measurable group with a left-invariant probability measure $\G$. Let $\rho:G \to \Aut V$ be a representation. Suppose that for all $v$ in $V$ the map $g \to \rho(g) v$ is measurable.
	
	If $S$ is a non-empty $G$-invariant convex subset of $V$,
	then $S$ contains a $G$-invariant element.
	\label{thm:genericMutualFundTheorem}
	
	If $G$ is a finite group, we only need require that $V$ is a vector space and $G$ acts by linear automorphisms.
\end{theorem}

The theorem is proved by taking an arbitrary element of the
set and then averaging over the action of the group.

For financial applications, we we may take $G$ to be a subgroup of the automorphism group of the market which admits a left-invariant density and $\rho$
to be the standard action of $G$ on $L^1(\Omega; \R)$. This allows us
to simplify invariant convex optimization problems by restricting attention to
invariant investment strategies.

We will see a number of applications of this general result throughout this paper. In this section we will use this result to prove the classical two-mutual-fund theorem of \cite{mertonMutualFund}. A similar argument was used in \cite{armstrongMarkowitz} to prove the classical two-mutual-fund theorem but the notion of isomorphism was different. Before proving the two-mutual-fund theorem
we will show how the notion of isomorphism in \cite{armstrongMarkowitz} relates to our new definition. We will do this by defining a general notion
of a ``finite-dimensional linear market'' and giving a classification result for
such markets and their isomorphisms.

\begin{definition}
	A one-period financial market $M=((\Omega, {\cal F},\P),c)$ is {\em separated} if there is a subset $\mathring{\Omega} \subset \Omega$ of full measure such that for any distinct
	$\omega_1$, $\omega_2 \in \mathring{\Omega}$ there exists $X \in \dom c$ with $X(\omega_1)\neq X(\omega_2)$.

	A one-period financial market is {\em linear} if $\dom c$ is a linear subspace
	of $L^0(\Omega; \R)$ and $c$ is linear on $\dom c$. The {\em dimension} of a linear market is the dimension of $\dom c$.
	
	On a linear market, we may define a map $\pi$ from $\Omega$ to $(\dom c)^*$, the algebraic dual space of $\dom c$, by
	\begin{equation}
	\pi(\omega)(X)=X(\omega)
	\label{eq:definitionOfPi}
	\end{equation}
	for $X \in \dom C$ and $\omega \in \Omega$. One checks that
	$\pi(\omega)(\alpha X_1+X_2)=(\alpha X_1+X_2)(\omega) = \alpha X_1(\omega)+X_2(\omega)=\alpha \pi(X_1) + \pi(X_2)$,
	so $\pi(\omega) \in (\dom c)^*$ as claimed. The map
	$\pi$ induces a sigma algebra and measure on $(\dom c)^*$. We write $d_M$ for this measure, which we call the {\em distribution}
	of the market. If $M$ is separated, then $\pi$ is a mod 0 isomorphism.
\end{definition}

Financially, a market is linear if all traded assets can be bought and sold in unlimited quantities at a fixed price per unit. A market is separated if the probability space contains no information other than that captured by asset prices.

A finite-dimensional real vector space has a natural topology defined by the requirement that linear isomorphisms to $\R^n$ are homeomorphisms.
We would like to require that the measure $d_M$ is in some sense compatible with this topology. To be precise we recall
the following definition.
\begin{definition}(see \cite{itoIntroduction})
	A {\em regular} probability measure is a probability
	measure arising as the Lebesgue extension of a Borel probability measure on a topological space.
\end{definition}	
We would like to be able to ensure that $d_M$ is a regular probability measure. To do this we require an additional condition
on the probability space  $(\Omega,\cal F,\P)$.

\begin{definition}(see \cite{rokhlin} and \cite{itoIntroduction})
	A probability space $(\Omega,\cal F,\P)$ is	{\em standard} if
	it is isomorphic mod $0$ to either: the Lebesgue measure on $[0,1]$;
	a probability space on a finite or countable number of atoms; a convex combination of both.
\end{definition}

The study of standard probability spaces was started by \cite{vonNeumann}. Although it may appear to be a highly restrictive condition, it is in fact a very mild assumption.
It\^o summarised the situation in \cite{itoIntroduction} as ``all probability spaces 
appearing in practical applications are standard''. We note a number of important examples that justify this claim. All regular probability measures on a complete separable
metric space are standard. This includes all regular measures on $\R^n$ and the Wiener measure on $C^0[0,\infty)$. Finite and countable products of standard spaces
are standard. A non-null measurable subset of a standard probability space becomes a standard probability space
when endowed with the conditional measure. For proofs of these assertions see \cite{rokhlin} or \cite{itoIntroduction}.

\begin{lemma}
	If $M$ is a finite-dimensional linear market based on a standard probability space,
	then $d_M \in \P((\dom c)^*)$ where $\P(S)$ denotes the set of regular probability measures on $S$. 
	\label{lemma:regularity}
\end{lemma}

\begin{definition}
	A regular probability measure on a finite-dimensional vector space, $V$, is said
	to be {\em non-degenerate} if for any $X,Y\in V^*$, $X=Y$ almost everywhere implies $X=Y$.
\end{definition}
Degenerate probability measures arise when the measure is concentrated on a vector subspace.
\begin{definition}
	$\VecM$ is defined to be the category with objects consisting of triples
	$(V,d,c)$ with $V$ a finite-dimensional vector space, $d \in \P(V)$ with $d$ non-degenerate and $c \in V$. $\VecM$ is equipped with a notion of morphism given by linear transformations $T:(V_1,d_1,c_1)\to(V_2,d_2,c_2)$ satisfying:
	\begin{enumerate}[nosep,label=(\roman*)]
		\item for any Borel measurable set $A \subseteq V_2$
		\begin{equation}
		d_2(A) = d_1(T^{-1}(A))
		\label{eq:distributionConditionForVecM};
		\end{equation}
		\item the vectors $c_1$ and $c_2$ are related by
		\begin{equation}
		c_2 = T(c_1).
		\label{eq:costConditionForVecM}
		\end{equation}
	\end{enumerate}
\end{definition}
\begin{definition}
	$\DualM$  is defined to be the category with objects consisting of triples $(V,d^*,c^*)$ with $V$ a finite-dimensional vector space, $d^* \in \P(V^*)$ with $d^*$ non-degenerate and $c^* \in V^*$.
	Morphisms $T:(V_1,d^*_1,c^*_1)\to (V_2,d^*_2,c^*_2)$ in $\DualM$ are given by a linear transformation $T:V_1 \to V_2$ whose
	whose vector space dual $T^*$ is a $\VecM$ morphism $T^*:(V_1^*,d^*_1,c^*_1)\to 
	(V_2^*,d^*_2,c^*_2)$.
\end{definition}	
\begin{definition}
	$\FinM$ is defined to be the category with objects given by separated finite-dimensional linear markets whose probability space is standard, and morphisms given by market morphisms. 	
\end{definition}

For any element $M$ of $\FinM$ define 
\[\Vec(M)=((\dom c)^*,d_M,c).\]
In the opposite direction, for any element $((V,d,c))$ of $\VecM$ we define
\[
\Fin((V,d,c))=((V,{\cal F},d),{\underline c})
\]
where ${\cal F}$ is the sigma algebra associated with $d$ and the map $\underline{c}:L^0(V;\R) \to \R$
satisfies
\[
\underline{c}(X)= \begin{cases}
X(c) & \text{if $X$ is equal to a linear map almost everywhere,} \\
\infty & \text{otherwise}. \\
\end{cases}
\]

\begin{theorem}[Equivalence of vector space and probabilistic categories of market]
	$\Vec(M)$ lies in $\VecM$ and the map $\Vec:\FinM \to \VecM$ defines a bijection on isomorphism classes. $\Fin((V,d,c))$ lies
	in $\FinM$. We may extend $\Vec$ and $\Fin$ to functors by defining their action on morphisms such that $\Vec$ and $\Fin$ define an equivalence of categories. Similarly the map
	$\Dual:\ob(\FinM)\to\ob(\DualM)$
	given by
	$\Dual(M)=(\dom c,d_M,c)$
	may be extended to a give a duality of the categories $\FinM$ and $\DualM$.
	\label{thm:linearmarkets}
\end{theorem}

To interpret this result financially, we suppose that we have a market
of $n$ assets. The space of portfolios in these assets is an $n$-dimensional
vector space $V$.
The cost of a portfolio defines a linear functional $c^*$ on this vector space.
The eventual payoff of a portfolio gives rise to a random linear functional acting on the space of portfolios.
The distribution of this payoff functional is given by $d^*$. Together
this data defines an element $(V,d^*,c^*) \in \ob(\DualM)$. Thinking of the space of portfolios as a vector space with no preferred basis represents the financial idea
that a portfolio of assets can be viewed as an asset in its own right.
The category $\DualM$ is therefore the appropriate category to use if one believes
that the distinction between an asset traded on the market and a portfolio of assets is not financially significant.

The significance of Theorem \ref{thm:linearmarkets} is that it shows the notion of equivalence of markets obtained by treating all portfolios as equally valid investment strategies is the same as the notion of equivalence given in Definition \ref{def:morphism}. This relates the definitions of \cite{armstrongMarkowitz} to the definitions in this paper. The advantage of our new Definition \ref{def:morphism}
is that it can be applied to infinite markets, as we shall see when we discuss complete markets later, and to non-linear markets.

The proof of Theorem \ref{thm:linearmarkets} shows that morphisms
in $\VecM$ are surjective linear transformations. It follows that the morphisms
of $\DualM$ are injective. This backs up the claim we made
earlier that market morphisms are a contravariant representation of market inclusion.

We now apply this general theory to the case of assets following a multivariate normal distribution, as considered by Markowitz \cite{markowitz}.

Let $g_{\mu}$ be the multivariate normal distribution with mean $\mu \in \R^n$ and covariance matrix given by the identity $\id_n$. We say that a market is Gaussian if it is isomorphic to a market on $\R^n$ with density $g_{\mu}$. Trivially any Gaussian market is isomorphic to a market of the form $\Fin(\R^n,g_{\mu},c)$ for some $\mu, c \in \R^n$. Let $\{e_i\}$ be the standard basis for $\R^n$. Since isometries of $\R^n$ preserve the Gaussian measure, we may apply a rotation so that $\mu$ lies in the span of $e_1$ and $c$ lies in the span of $e_1$ and $e_2$. This shows that any Gaussian market can be written in the form
\begin{equation}
\Fin( \R^n, g_{\alpha\, e_1}, \beta \, e_1 + \gamma \, e_2 ), \quad \alpha,\beta,\gamma \in \R.
\label{eq:genericMarkowitz}
\end{equation}

We now have the following classification theorem.
\begin{theorem}[Classification of Markowitz markets]
	Let $M \in \FinM$ be a market and suppose that $\{X_i\}$ is a basis for $\dom c$
	given by assets following a multivariate normal distribution.  Then $M$ is Gaussian,
	and hence is isomorphic to a market of the form \eqref{eq:genericMarkowitz}
	\label{thm:markowitzClassification}.
\end{theorem}
\noindent This theorem is essentially a restatement of the main classification result of \cite{armstrongMarkowitz} in the language
of one-period markets.
\begin{corollary}
	All invariant investment strategies $X \in \dom c$ in a Gaussian market
	lie in a two-dimensional vector subspace of $\dom c$.
	\label{cor:generalMutualFund}
\end{corollary}
\begin{corollary}
	(Two-mutual-fund theorem \cite{mertonMutualFund})
	Suppose we have $n$ assets of a given cost whose payoffs follow a multivariate normal distribution. We wish to find the portfolio of assets with minimum variance but with a given expected payoff $C_1$ and cost $C_2$. There are two portfolios $X_1$ and $X_2$ independent of $C_1$ and $C_2$ such that we can solve these mean--variance optimization problems for any $C_1$ and $C_2$ simply by considering linear combinations of $X_1$ and $X_2$.
\end{corollary}
\noindent The portfolios $X_1$ and $X_2$ are the two ``mutual funds'' that
give this theorem its name.

We remark that Corollary \ref{cor:generalMutualFund} is a much stronger result than the classical two-mutual-fund theorem. The paper \cite{armstrongMarkowitz} gives numerous concrete examples of financially interesting results arising from invariance arguments other than just the two-mutual-fund-theorem.

We also remark that the concrete isomorphism found in Theorem \ref{thm:markowitzClassification} makes it extremely easy to solve the classical mean-variance optimization problem directly, thereby recovering the full set of results found in \cite{mertonMutualFund}. This approach is pursued in \cite{armstrongMarkowitz}.

\section{One-period complete markets}
\label{sec:completeOnePeriod}

\begin{definition}
	A one-period market $M=((\Omega,{\cal F},\P),c)$ is {\em complete} if there exists
	a measure $\Q$ on $\Omega$ equivalent to $\P$, and $C > 0$ such that
	\begin{equation}
	c(X) = \begin{cases}
	C (\E_\Q(X^+) - \E_\Q(X^-) ) & \text{one of $\E_\Q(X^{\pm})$ is finite} \\
	\infty & \text{otherwise}. \\
	\end{cases}
	\label{eq:defQ}
	\end{equation}
	In this formula $X^+$ and $X^-$ denote the positive and negative parts of the random variable $X$.
	We note that $c(1)=C$, so we interpret $(C-1)$ as a deterministic interest rate.
\end{definition}

\begin{example}
	Let $I$ be the market given by taking the $\P$ and $\Q$ measure to both be equal to the Lebesgue measure on $[0,1)$ and with cost of the constant function with value $1$, equal to $1$. In this market prices are given by expectations, so we call $I$ a {\em casino}. (Our casino is of course an idealized one, in which the profits and losses of a typical client form a martingale rather than a supermartingale.)
\end{example}

Given a complete market $M$
we may define a new complete market $M \times I$ by taking the product measures for both the $\P$ and the $\Q$ measures and taking the constant $C$ to be that given by the market $M$.

From a financial point of view the market $M \times I$ represents the market obtained by considering investment strategies where one first invests in the market $M$ and then places a bet at the casino.

In applications it is not unreasonable to assume that there is a casino available should a trader wish to use it. So classifying complete markets of the form $M \times I$ should be just as useful in practice as a full classification.
The theorem below gives a classification for markets of this form.
\begin{theorem}[Classification of complete markets up to a casino]
	Let $M$ be a complete market on a standard probability space. Then $M \times I$ is isomorphic to $\tilde{M} \times I$,
	where $\tilde{M}$ is the market with probability space given by
	$\tilde{\Omega}=[0,1]$ equipped with the Lebesgue measure and
	with pricing function
	\[
	\tilde{c}(X) = C \int_0^1 F^{-1}_\frac{\ed \Q}{\ed \P} X(x) \ed x.
	\]
	\label{thm:simpleCompleteMarket1}
	Here $F^{-1}_{\frac{\ed \Q}{\ed \P}}$ is the inverse distribution function
	of $\frac{\ed \Q}{\ed \P}$ on $M$.
\end{theorem}

The first step toward proving this is to observe that we may recover $\Q$ from $c$ since for any measurable set $A \subset \Omega$ we have
\[
\Q(A) = \E_\Q(1_A) = \frac{c(1_A)}{c(1)}.
\]
It follows that two one-period complete markets $((\Omega_i,{\cal F}_i,\P_i),c_i)$ ($i=1,2$) are isomorphic if and only if (a) there is a mod 0 isomorphism for the $\P_i$ measures which is also a mod 0 isomorphism for the $\Q_i$ measures; and (b) the cost of the constant function with value $1$ is equal in both markets.

There may be more than just $2$ measures on the market which are of financial interest. A trader with views about the market represented by a measure $\P$ may be constrained by a risk manager or regulator with different views about the market. These can be represented by alternative measures. Let us state a classification result
similar to
Theorem \ref{thm:simpleCompleteMarket1} that applies to this situation.

\begin{theorem}[Classification of complete markets with multiple views]
	Let $I$ denote the interval $[0,1)$ with the Lebesgue measure. We suppose that $\P_0, \P_1, \ldots, \P_n$
	are equivalent probability measures on $(\Omega, {\cal F})$. We assume $\P_0$ is standard. Then there is a unique Lebesgue measure $\P_0^\prime$ on
	$\Omega^\prime=(0,\infty)^n$ such that $\P_0 \times I$ and $\P_0^\prime \times I$ are mod 0 isomorphic via an isomorphism which also acts as a mod 0 isomorphism between the measures $\P_i \times I$ and $\P_i^\prime \times I$ where $\P^\prime_i$ is the Lebesgue measure given by 
	\[
	\P^\prime_i(A)=\int_{(0,\infty)^n} \omega_i 1_A( \omega) \, \ed \mu.
	\]
	In this formula, $A$ is a measurable set, $1_A$ is the indicator function $A$ and $\omega_i$ is the $i$-th coordinate function on $\R^n$. Note that we must have $\E_{\P_0^\prime}(\omega_i)=1$ for these $\P^\prime_i$ to be probability measures.
	\label{thm:simpleCompleteMarketN}
\end{theorem}

We note the following financial implication (using the notation of Theorem \ref{thm:simpleCompleteMarketN}).

\begin{corollary}[Convex mutual-fund theorem for complete markets]
	Let $A$ be a non-empty convex subset of the space
	of $\P_0$-integrable random variables on $\Omega$. Suppose that $A$ is also invariant under mod 0 isomorphisms that preserve all the $\P_i$. Then $A$ contains an element which can be written as a function of the Radon--Nikodym derivatives $\frac{\ed \P_i}{\ed \P_0}$. 
\label{cor:convexMutualFund}
\end{corollary}

For example, $A$ might arise as the optimal investment strategies in a convex optimization problem with a cost constraint and risk-management constraints imposed by a number of regulators and risk managers given in terms of the $\P_i$.

A special case of the result above is the problem of expected-utility optimisation in a complete market subject to a single cost constraint for a concave, increasing utility function. In this case it is well-known that the optimal investment has a payoff function given as a function of the Radon--Nikodym derivative (see \cite{follmerSchiedBook}).

\medskip

Let us now give the definitions needed to state a full classification for complete
one-period markets.
Write ${\cal S}$ for the set
of mod 0 isomorphism classes of standard probability spaces. We call ${\cal S}$ the moduli space of standard probability spaces.

Given $m \in {\cal S}$, we define $m_0$ to be the measure of the continuous component of $m$ (or zero if it has no continuous component) and we define $m_i$ for $i>0$ to be the measure of the $i$-th largest atom in our probability space (or $0$ if there less than $i$ atoms). Thus we have identified a correspondence between ${\cal S}$
and sets of numbers $m_i$ ($i \in \N$) which satisfy
\begin{equation}
m_i \in [0,1]; \qquad
\forall i \in \N^+, \, m_i \geq m_{i+1}  ; \qquad \text{and }
m_0 = 1-\sum_{i=1}^\infty m_i.
\label{eq:defCalS}
\end{equation}
We give ${\cal S}$ the topology induced by thinking of it as a subset of $\R^\infty$ in this way. Thus we may talk about measurable maps to ${\cal S}$, or ${\cal S}$-valued random variables.

The theory of disintegration of measure tells us that for a complete market $M$ based on a standard probability space, there is a $\mu_M$-almost-surely unique measurable function
\[
m_M:(0,\infty) \to {\cal S}
\]
with $m_M(x)$ given by the mod 0 isomorphism class of the $\P$ conditional measure conditioned on the value of
$\frac{\ed \Q}{\ed \P}=x$ and where $\mu_M$ denotes the measure on $(0,\infty)$ induced by $\frac{\ed \Q}{\ed \P}$.

\begin{definition}
	Let $\Measures(n)$ be the set consisting of pairs $(\mu,m)$ where:
	\begin{enumerate}[nosep,label=(\roman*)]
		\item $\mu$ is a regular probability
		measure on $(0,\infty)^n$ satisfying
		$\E_{\mu}( \omega_i ) = 1$
		for the $i$th coordinate function $\omega_i$ on $\R^n$;
		\item $m$ is an ${\cal S}$ valued $\mu$ random variable.
	\end{enumerate}
\end{definition}

\begin{theorem}[Generalised classification of complete markets]
	Standard probability spaces $(\Omega, {\cal F}, \P_0)$ equipped with $n$-additional
	equivalent measures $\P_1$, \ldots, $\P_n$ are classified up to joint 
	$\P_0$-, \ldots, $\P_n$- mod 0 isomorphism by elements $(\mu_q,m_q) \in \Measures(n)$.
	Here $\mu_q$ is the measure on $(0,\infty)^n$ induced by the $\R^n$ vector valued function $q$
	with $i$-th component given by the Radon--Nikodym derivative $\frac{\ed \P_i}{\ed \P_0}$.
\label{thm:classificationComplete}
\end{theorem}

The proof uses Rokhlin's theory of the decomposition of measure.

\subsection{Non-convex problems and rearrangement}
\label{sec:rearrangement}

We show in this section that Theorem \ref{thm:simpleCompleteMarketN} allows us to identify
a mutual-fund theorem that applies to optimization in complete markets when we assume that the problem is ``monotonic''
rather than convex.

We have in mind applications to behavioural economics
based on the observations of Kahneman and Tversky in \cite{kahnemanTversky}. For examples of applications of Kahneman and Tversky's ideas to mathematical finance and risk management, see, for example, \cite{xyzFirst}, the review \cite{xyzReview}, 
and \cite{armstrongBrigoSShaped} which contains numerous further references.

It has been observed in this literature (see for example \cite{xyzPortfolioChoiceViaQuantiles}) that
the solution to optimal investment problems
in complete markets involving S-shaped utility functions can be obtained by considering monotonic functions of the Radon--Nikodym derivative $\frac{\ed \Q}{\ed \P}$. The aim of this section is to show how these
results arise from general monotonicity properties, automorphism invariance
and our classification theorems. We take the opportunity to show how these results can be generalized to situations where there are more than two measures $\P$ and $\Q$, for example, to the case where risk managers and traders have different beliefs about the future evolution of the market.

Given two random variables $X$, $Y$ on a probability space $(\Omega, {\cal F}, \P)$
we write
\[
d^\P(X) \preceq d^\P(Y)
\]
if $F_X(k):=\P(X \leq k) \geq \P(Y \leq k)=:F_Y(k)$ for all $k$. The notation $d^{\P}(X)$ is intended to suggest ``the $\P$-distribution of $X$''. Given a third random variable $Z$ we write
\[
d^\P(X \mid Z) \preceq d^\P(Y \mid Z)
\]
if $\P(X \leq k \mid Z) \geq \P(Y \leq k \mid Z)$ almost surely for all $k$.

We suppose that market participants such as traders and risk managers impose
some form of relation $\preceq^\prime$ on random variables to express their preferences between different investment opportunities. One might reasonably expect that
\begin{equation}
X \preceq Y \implies X \preceq^\prime Y.
\end{equation}
If this condition holds, we will say that $\preceq^\prime$ is {\em increasing}. We say that $\preceq^\prime$ is {\em decreasing} if the reversed relation is increasing. We say that a relation on random variables is {\em monotonic} if it is either increasing or decreasing. We say that the {\em sign} of a monotonic relation is $1$ if it increasing or $-1$ if it is decreasing.

\begin{definition}[Rearrangement]
	\label{def:rearrangement}	
	Let $m$ be a Lebesgue probability measure on $(0,\infty)$. Let $F_m$ denote
	the cumulative distribution function of $m$. Write $x, y$ for the coordinate functions on $(0,\infty)\times[0,1)$. Define $U_m:(0,\infty)\times[0,1)\to [0,1]$ by
	\[
	U_m(\omega) = (1-y(\omega)) \lim_{x^\prime\to x(\omega)-} F_m(x^\prime) + y(\omega)\lim_{x^\prime\to x(\omega)+} F_m(x^\prime).
	\]	
	$U_m$ is well-defined since $F_m$ is c\`adl\`ag.
	We write $\P_m$ for the product measure on $(0,\infty)\times[0,1)$. If
	\begin{equation}
	\E_{\P_m}( x(\omega) ) = 1
	\label{eq:conditionOnM}
	\end{equation}
	then $x$ is the Radon--Nikodym derivative of an equivalent measure we call $\Q_m$.
	Given $X \in L^0_{\P_m}((0,\infty) \times [0,1]; \R)$ we define
	the {\em increasing and decreasing rearrangements} of $X$ by
	\[
	R^+_m(X) = F_X^{-1}(U_m), \quad R^-_m(X) = -F_{-X}^{-1}(U_m)
	\]
	respectively, where $F^{-1}_X$ is the $\P_m$ inverse distribution function of $X$.
\end{definition}

Our next theorem shows that the notion of rearrangement can be generalized to situations when
there are more than two probability measures under consideration.

\begin{theorem}[Monotone mutual-fund theorem for complete markets]
	Let $(\Omega,{\cal F},\P_0)$ be a standard probability space equipped with
	$n$ equivalent measures $\P_i$ ($1 \leq i \leq n$). Let $I=[0,1)$. Let $\preceq_i$ ($1 \leq i \leq n$) be monotonic relations on the set of probability distributions on $\R$. Write $\sign i$ for the sign of $\preceq_i$. There exists a mapping $R:L^0(\Omega \times I)\to L^0(\Omega \times I)$, which we call {\em rearrangement}, with the following properties.
	\begin{enumerate}[nosep,label=(\roman*)]
		\item Rearrangment does not change $\P_0$ distributions:
		\[
		d^{\P_0}(X) = d^{\P_0}(R(X)).
		\]
		\item Rearrangement increases or decreases $\P_i$ distributions according to the sign of $\preceq_i$:
		\[
		d^{\P_i}((\sign i) X) \preceq d^{\P_i}((\sign i) R(X)), \quad 1 \leq i \leq n.
		\]
		\item Let $q$ denote the vector of $n$ Radon--Nikodym derivatives $\frac{\ed \P_i}{\ed \P_0}$.
		Define $\preceq$ on $\R^n$ by $x \preceq y$ if $(\sign i)x_i \leq (\sign i)y_i$ for all components $i$, and hence
		define $\prec$ on $\R^n$. Then $R(X)$ satisfies
		\[
		R(X)(\omega) \leq R(X)(\omega^\prime)
		\quad \text{ if } \quad q(\omega) \prec
		q(\omega^\prime).
		\]
	\end{enumerate}
	\label{thm:rearrangement}
\end{theorem}

This theorem gives a general structural theorem about optimal investments in complete markets containing a casino. So long as the optimality criterion and any pricing or risk constraints are monotonic in some measures $\P_i$, we can restrict our attention to strategies that lie in the image of $R$. We interpret this as a mutual-fund theorem since it says that, for a general class of optimization problems, we can safely restrict attention to a subset of the random variables available in the market.

The assumption that there is a casino can be dropped in many
cases since, as one might intuitively expect, one often doesn't take any real advantage of the casino. This is formalized in the next corollary.

\begin{corollary}
	Let $(\Omega, {\cal F}, \P_i)$ ($1 \leq i \leq n)$ be as in the previous
	Theorem \ref{thm:rearrangement}
	We can find a map $\tilde{R}:L^0(\Omega)\to L^0(\Omega)$ which shares properties (i), (ii) and (iii) described in Theorem \ref{thm:rearrangement} so long as either:
	(a)
		\item $\P_0$ is atomless and $n=1$;
	or (b) for some $j$, the distribution of $\frac{\ed \P_j}{\ed \P_0}$	
		conditioned on the value of all the other Radon--Nikodym derivatives
		is almost surely continuous. In case (b)
		$\tilde{R}$ can be assumed to depend only on the
		value of  $q$.
	\label{cor:rearrangement}
\end{corollary}
Note that the theory
of conditional distributions detailed in \cite{itoIntroduction}
ensures that the conditional distribution exists in case (b).

\section{Continuous-Time Markets}
\label{sec:ctstime}

Let us extend our definitions of markets to the multi-period setting.

\begin{definition}
A multi-period market consists of the following.
\begin{enumerate}[nosep,label=(\roman*)]
	\item A filtered probability space
	$(\Omega, {\cal F}_t, \P)$ where $t \in {\cal T} \subseteq [0,T]$ for
	some index set ${\cal T}$ containing both $0$ and $T$.
	We write ${\cal F}={\cal F}_T$. We require ${\cal F}_0=\{\emptyset,\Omega\}.$
	\item For each $X \in L^0(\Omega;\R)$, an ${\cal F}_t$-adapted process $c_t(X)$ defined for $t$ in ${\cal T}\setminus T$.
\end{enumerate}
\end{definition}

Random variables $X \in L^0(\Omega, {\cal F}_T; \R)$ are interpreted as
contracts which have payoff $X$ at time $T$. The cost of this contract at time $t$ is $c_t(X)$.

We note that this is deliberately bare-bones definition of a market. In practice would want to impose additional conditions on the $c_t$. For example, one would normally wish to forbid arbitrage and to impose ``the usual conditions'' on the filtered probability space.

\begin{definition}
	A {\em filtration isomorphism} of filtered spaces $(\Omega, {\cal F}, {\cal F}_t, \P)$ where $t \in {\cal T}$ for some index set ${\cal T}$ is a mod $0$ isomorphism for ${\cal F}$ which is also a mod $0$ isomorphism for each ${\cal F}_p$.
	An {\em isomorphism} of multi-period markets is a filtration isomorphism that preserves the cost functions.
\end{definition}

Given a one-period market $((\Omega, {\cal F}, \P),c)$ we can trivially define a filtration ${\cal F}_0=\{\emptyset, \Omega\}$, ${\cal F}_1={\cal F}$ indexed by $\{0,1\}$ and we may define $c_0=c$. Hence we can define a multi-period market in a canonical fashion from a one-period market. The notion of isomorphism is preserved. In this sense, our definition of multi-period markets and their isomorphisms is a 
generalization of the corresponding notions for one-period market.

\begin{definition}[Exchange market]
	\label{def:exchangeMarket}
	Let $(\Omega, {\cal F}_t, \P)$ be $n$-dimensional Wiener space, that is  the probability space generated by the $n$-dimensional
	Brownian motion $\bm{W}_t$. Let $\bm{X}_t$ be an $n$-dimensional
	stochastic processes defined by a stochastic differential equation of the form
	\begin{equation}
	\ed {\bm{X}}_t = \bm{\mu}(\bm{X}_t,t) \, \ed t + 
	\bm{\sigma}(\bm{X}_t,t) \, \ed {\bm{W}}_t.
	\label{eq:nDDiffusion}
	\end{equation}
	Here $\bm{\mu}$ is an $\R^n$-vector valued function and $\bm{\sigma}$
	is an invertible-matrix valued function. We assume the coefficients $\bm{\mu}$ and $\bm{\sigma}$ are sufficiently well-behaved for the solution of the equation to be well-defined on $[0,T]$. The components, $X^i_t$, of the vector $\bf{X}_t$ are intended to model the prices of $n$-assets.
	
	The {\em exchange market} for $\eqref{eq:nDDiffusion}$ with risk-free rate $r$
	over a time period $[0,T]$ is given by defining $c_t:L^0(\Omega;\R)\to \R$ for $t \in [0,T)$ by
	\begin{subnumcases}{c_t( X ) = }
	\alpha_0 \, e^{-r(T-t)} + \textstyle \sum_{i=1}^n \alpha_i \, X^i_t &
	\text{if $X=\alpha_0 + \sum_{i=1}^n \alpha_i \, X^i_T$}, \label{eqn:bsmCommon}	\\
	\infty & \text{otherwise}.
	\label{eqn:bsmBuyAndHold}	
	\end{subnumcases}	
	This is well-defined so long as we assume that $X^i_T$ are linearly independent random variables. This will be the case in all situations of interest.
\end{definition}

The market defined above is called an exchange market because it models the
basic assets that can be purchased directly on an exchange, but does not
take into account the possibility of replicating payoffs via hedging. The next
definition does take this into account.

\begin{definition}[Superhedging market]
	\label{defn:superhedgingMarket}
	The {\em superhedging market} for $\eqref{eq:nDDiffusion}$ with risk-free rate $r$
	over a time period $[0,T]$ is given by defining $c_t(X)$ to be the infimum of the
	cost at time $t$ of self-financing trading strategies that superhedge $X$.
	See \cite{harrisonPliska} for a definition of a self-financing trading strategy.
	A self-financing trading strategy superhedges $X \in L^0(\Omega;{\cal F}_T)$ if the final payoff of the strategy is always greater than or equal to $X$.
\end{definition}

Thus the superhedging market represents the effective market of derivatives
that a trader can achieve given the exchange market. The cost function $c_t$ for such a market
is the superhedging price. Of particular interest are complete markets where any contingent claim may
be both superhedged and subhedged. One expects that the price
in an arbitrage-free market can be expressed as a risk-neutral probability. These remarks motivate the next
definition.

\begin{definition}
	A continuous-time market $(\Omega, {\cal F}_t, \P),c_t)$ on $[0,T]$ is called a {\em continuous-time complete market with risk-free rate $r$} if there exists a measure $\Q$ equivalent to $\P$ with
	\begin{equation}
	c_t(X) = e^{-r(T-t)} \E_\Q( X \mid {\cal F}_t)
	\label{eq:costFunctionInQ}
	\end{equation}
	for $\Q$-integrable random variables $X$ and equal to $\infty$ otherwise. We follow our usual conventions on expectations to allow $-\infty$ when the positive part of an expectation is finite and the negative part is infinite.
\end{definition} 

Using our new terminology, the theory of Harrison and Pliska
\cite{harrisonPliska} shows how the superhedging market associated with the
SDE \eqref{eq:nDDiffusion} gives rise to a continuous-time complete market, subject to sufficient  regularity assumptions on the coefficients.

\begin{definition}
	The continuous-time complete market with risk-free rate and cost function given
	by the superhedging market is called the {\em complete market associated with the SDE \eqref{eq:nDDiffusion}} (subject to the required regularity assumptions for
	$q_t$ to be a well-defined $\P$-martingale).
\end{definition}

We differ slightly in our presentation from Harrison and Pliska \cite{harrisonPliska} in that they discuss replication and we consider superhedging. This is why we are willing to ascribe a cost of $-\infty$ to some $X \in L^0(\Omega,{\cal F}_T)$, whereas if one insists on replication, $X$ must be absolutely integrable. The definition of the superhedging market associated to a given market can be applied equally well to incomplete markets where there is a more meaningful difference between replication and superhedging. This is why we prefer to think in terms of superhedging, and in this we are influenced by the presentation of \cite{pennanenDuality}.

\begin{definition}
	The {\em absolute market price of risk} in the complete market associated with the SDE \eqref{eq:nDDiffusion} is 
	the element of $L^0(\Omega \times [0,T], \P)$ defined by
	\[
	\AMPR_t = | \bm{\sigma}^{-1}(r \bm{X}_t -  \bm{\mu}) |.
	\]
\end{definition}

\begin{theorem}
	Let $F$ be the contravariant functor mapping a continuous-time market, $M$ with underlying probability space $\Omega$ to
	the vector space $L^0(\Omega \times [0,T], \P \times \lambda)$ where $[0,T]$ is equipped with the Lebesgue measure $\lambda$, and where $F$ acts on morphisms
	$\phi:\Omega_1 \to \Omega_2$ by
	$
	F(\phi)(X)=X \circ( \phi \times \id)
	$
	for $X \in L^0(\Omega_2 \times [0,T], \P \times \lambda)$. We recall that elements of $L^0(\Omega \times [0,T], \P \times \lambda)$
	are defined to be almost-sure equivalence classes. Write
	$\AMPR(M) \in L^0(\Omega \times [0,T], \P \times \lambda)$ then for any
	market isomorphism $\phi$
	\[
	\AMPR(\phi(M))=F(\phi^{-1}) \AMPR(M).
	\]
	We summarize this by saying that the absolute market price of risk is an invariantly-defined element for $F$.
	\label{thm:marketPriceOfRiskInvariant}
\end{theorem}

Theorem \ref{thm:marketPriceOfRiskInvariant} can be viewed as an analogue of Gauss's Theorema Egregium for the category of continuous-time complete markets. Of course, we are only claiming that this is an analogy. We have not
established any relationship between markets and Gaussian curvature. If one is interested in direct
relationships between curvature and finance, one can consider the theory of SDEs on manifolds, the Riemannian
metric defined by a non-degenerate volatility term and the corresponding curvature tensor (see, for example, \cite{henryLabordere}). Note that the Riemannian metric arising in this way is independent of the choice of drift term,
and so one may have non-zero curvature even when $\P=\Q$.

The proof of Theorem \ref{thm:marketPriceOfRiskInvariant} suggests we extend the definition of $\AMPR_t$ to all complete markets as follows.
\begin{definition}	
	In a continuous-time complete market we define $\AMPR_t \in L^0_{\geq 0}(\Omega \times \R)$ (if it exists)
	to be the solution of
	\begin{equation}
	\int_0^t \frac{1}{Q^2_t} \, \ed [Q, Q]_t := \int_0^t \AMPR^2_t \, \ed t,
	\label{eq:amprDefGeneral}
	\end{equation}
	where
	\begin{equation}
	Q_t:= \frac{\ed \Q}{\ed \P} \Big|_{{\cal F}_t}.
	\label{eq:qtDefGeneral}
	\end{equation}
\end{definition}

An additional invariant we need to consider is the dimension of our market. This is given by the number of independent
Brownian motions $n$. Our next result shows that this is an invariant of the market; indeed it is an invariant of the
filtered probability space $(\Omega, {\cal F}_t,\P)$.

\begin{definition}
	The $n$-dimensional {\em Wiener space} on $[0,T]$ is the filtered probability space generated by $n$ independent standard Brownian motions on $[0,T]$. A
	filtered probability space is called a Wiener space if it is filtration isomorphic to an $n$-dimensional Wiener space.
\end{definition}	

\begin{theorem}
	The dimension of a Wiener space is invariant under filtration isomorphisms.	\label{thm:invarianceDimension}
\end{theorem}

We are now ready to state a classification theorem for complete markets with deterministic absolute market price of risk. 
We recall that $\{e_i\}$ is the standard basis for $\R^i$ and $\id_n$ is the identity matrix. 

\begin{theorem}[The test case]
	Let $M$ be a continuous-time complete market with risk-free rate $r$, time period $T$
	based on a Wiener space of dimension $n$ and with $\AMPR$ given by
	\[
	\AMPR_t = A(t) \geq 0
	\]	
	for a bounded measurable function of time $A(t)$. Suppose
	that the process $q_t$ is continuous.
	In these circumstances $M$ is isomorphic
	to the complete market associated with the SDE \eqref{eq:nDDiffusion} with
	\[
	\bm{\mu}=r \bm{X}_t + A(t) \, e_1, \quad \text{and }
	\bm{\sigma}=\id_n
	\]
	and $\bm{X}_0=0$.
	\label{thm:testcase}
\end{theorem}
We call markets of this form {\em canonical Bachelier markets}.

The key step in the proof of this theorem is to invariantly define a Brownian motion
corresponding to the first component of ${\bm W}_t$. To do this, one shows that
\[
-\int_0^t \frac{1}{A(s)} \ed (\log Q)_s
\]
is a Brownian motion using Levy's characterisation of Brownian motion.

We have called Theorem \ref{thm:testcase} ``the test case'' as it is an analogous result to the theorem in differential geometry that a Riemannian manifold with vanishing curvature is flat.  This latter result is called ``the test case'' in \cite{spivak}.

\begin{example}
	The $n$-dimensional Black--Scholes--Merton market is isomorphic to a Bachelier market,
	since market price of risk in the Black--Scholes--Merton market is a deterministic
	constant vector.
\end{example}

\begin{example}
	Given a positive real number $A$, an invertible matrix $\bm{\sigma}$ and a vector $\bm{X}$, the set	of vectors $\bm{\mu}$ satisfying $|\bm{\sigma}^{-1}(r \bm{X} - \bm{\mu})|=A$ is non-empty; indeed, it as an ellipsoid. Hence given a complete continuous-time market modelled by an SDE, we may modify the drift to obtain a market 
	isomorphic to a Black--Scholes--Merton market with market price of risk $A$.
\end{example}

It is difficult to estimate the drift of a volatile asset. As a result,
the functional form of the drift is usually chosen for parsimony; one then
uses long-term data to calibrate this functional form. If one is following this approach,
in the absence of statistical evidence to the contrary it might be be reasonable to choose the functional form of the drift to ensure that
the resulting model has a constant market price of risk, and hence is isomorphic to a Black--Scholes--Merton model.

Our result shows that the many financial results that have been proved for the Black--Scholes--Merton model can be applied to a far wider range of markets
than one might at first sight expect. Even markets which seem superficially very different from the Black--Scholes--Merton market, such
as stochastic-volatility models, may still provide isomorphic investment
opportunities.

These observations suggests that one should separate optimal investment
problems into two components. One has the {\em strategic} problem
of optimal investment for a particular isomorphism class of market. Additionally one has the {\em tactical} problem of finding a concrete realisation (or approximate realisation) of the strategy, which can be interpreted as the task of finding a concrete morphism. This division of investment problems into strategic and tactical problems is already widely used in practice (see \cite{campbellViceira}).

Although we have restricted ourselves to considering markets with 
deterministic absolute market price of risk, this approach can be generalized.
Rather than attempt to model asset price dynamics directly, one may choose
a market model by attempting to model invariantly-defined quantities. For
example, if one has a view on the dynamics of the absolute market price of risk, one may develop a market model to reflect this. We expect this approach to yield
a systematic method for developing low-dimensional (and hence numerically tractable)
market models which still capture the essential features of the market. We
will explore this in future research.

\medskip

As an application of our classification theorem, we may now prove
a mutual-fund theorem.
\begin{theorem}[Continuous-time one-mutual-fund theorem]
	Let M be a complete continuous-time market with continuous $q_t$
	and with deterministic, bounded absolute market price of risk. Let $X^i_t$ for
	($1 \leq i \leq n$) be a collection of square integrable stochastic
	processes representing $n$ basic assets, then there exist $n$
	predictable real valued processes $\alpha^i_t$ such that any invariant, 
	non-empty, convex
	set of martingales contains an element which can be replicated by a continuous-time trading strategy using only the risk-free asset and the portfolio consisting of $\alpha^i_t$ units of asset $X^i_t$.
	
	In complete markets arising from SDEs of the form \eqref{eq:nDDiffusion}
	which also have a deterministic absolute market price of risk,
	we may take the portfolio $\bm{\alpha}$ with components $\alpha_i$ to be given by the vector
	\[
	(\bm{\sigma \sigma^\top})^{-1}(r \bm{X}_t - \bm{\mu}).
	\]
	\label{thm:ctsTimeMutualFundTheorem}
\end{theorem}
We note that a convex set of martingales
can be interpreted as a convex set of self-financing trading strategies or as a convex set of derivative securities.

We call this result a one-mutual-fund theorem because it shows that a fund
manager can create a single mutual fund that can be used to implement these
trading strategies. A key difference between our result and the classical one-mutual-fund theorem is that an investor needs to trade in our mutual
fund in continuous time.

This result explains the general form of the solution to the portfolio optimization problem studied by Merton in \cite{mertonPortfolio}. However, it goes considerably beyond this.

As an example, consider the problem of managing the investment and pension payments for a collective pension. Suppose
that the fund is heterogenous, so each individual may have a distinct mortality distribution, initial wealth and risk appetite. Assume that fund may invest in a Black--Scholes--Merton market, and that the individuals preferences and mortality are independent of this market. Assume that the investors preferences are convex. Our theorem now shows
that one need only consider investments in the risk-free asset and the mutual fund we have identified when deciding how to manage the pension. We can say this
without actually formulating an optimal investment problem describing how such a heterogeneous fund should be managed.

\medskip

In summary, our classification theorems have identified interesting isomorphisms between markets that are not obviously related. We have found large families of automorphisms for the classical markets of Markowitz and Black--Scholes--Merton. We have shown that considering these automorphisms allows one to prove very general mutual-fund theorems.

\section{Acknowledgements}

I would like to thank both the anonymous referees and Andrei Ionescu for their useful comments and corrections, and to thank Markus Riedle and Nick Bingham for their valuable advice.

\section{Funding}

I received no funding for this study.

\bibliographystyle{plain}
\bibliography{classifyingmarkets}

\begin{appendices}

\appendix

\section{Proofs}

\renewcommand{\thesubsection}{\thesection.\arabic{subsection}}

\subsection{Proofs for Section \ref{sec:markowitz}}

\begin{proof}[Proof of Lemma \ref{lemma:mod0Lemma}]
	Let $f:\Omega_1 \to \Omega_2$ be a $\Prob$ isomorphism. We can then find a homomorphism $g:\Omega_2 \to \Omega_1$ such
	that $g \circ f = \id_1$ almost surely and $f \circ g = \id_2$ almost surely. Define $\Omega_1^\prime$ to be the
	set of points where $fg(x)=x$ and $\Omega_2^\prime$ to be the set of points where $gf(y) = y$. $\Omega_1^\prime$ and $\Omega_2^\prime$
	will be of full measure.
	If $x_1, x_2 \in \Omega_1$ and $f(x_1)=f(x_2)$, then
	$gf(x_1)=gf(x_2)$, hence $x_1=x_2$. Thus $f$ is injective on $\Omega^\prime_1$. If $y \in \Omega^\prime_2$ then $fg(y)=y$, so
	$gfg(y)=g(y)$ and hence $g(y)\in \Omega_1^\prime$ with $y=fg(y)$. Thus $f$ maps $\Omega_1^\prime$ onto $\Omega_2^\prime$. Hence $f$
	is a mod 0 isomorphism.
	
	The converse follows trivially from the definitions.
\end{proof}

\begin{proof}[Proof of Lemma \ref{lemma:isomorphism}]
	Let $\phi:\Omega_1 \to \Omega_2$ be a $\Prob$ morphism with two-sided inverse $\phi^{-1}$.
	Suppose $\phi$ is, moreover, a market isomorphism.
	Using the fact that $\phi$ and $\phi^{-1}$ are both market morphisms, we have that for any $X\in L^0(\Omega_2; \R)$ we have
	\[
	c_2(X) = c_2(X \circ \phi \circ \phi^{-1})  \leq c_1(X \circ \phi) \leq c_2(X).
	\]
	Hence we must have equality throughout. Hence $c_2(X)=c_1(X\circ \phi)$.
	
	The result now follows from Lemma \ref{lemma:mod0Lemma}.
\end{proof}

\begin{proof}[Proof of Theorem \ref{thm:genericMutualFundTheorem}]
	Given $h \in G$, define $\phi_h:G \to G$ by left multiplication, so $\phi_h(g)=hg$. Let $A$ be a measurable set and let $1_A$ denote
	the indicator function of $A$ then
	\[
	1_A \circ \phi_h = 1_{h^{-1} A}.
	\]
	We deduce that
	\begin{equation}
	\E(X \circ \phi_h)=\E(X)
	\label{eq:rvinvariance}
	\end{equation}
	if $X$ is an indicator function of a set, and hence this holds for all integrable random variables $X$.
	
	By assumption $S$ is non-empty, so we may choose an element $s^\prime \in S$.	We define a random variable $X:G \to V$ by
	\begin{equation}
	X(g)=\rho(g) s^\prime.
	\label{eq:defnOfX}
	\end{equation}
	Because $G$ acts by isometries on $V$, $\|X(g)\|=\|\rho(g) s^\prime\|=\|s^\prime\|$ for all $g$. Hence by the dominated convergence theorem
	we may define an element $s$ by
	\begin{equation}
	s:=\E_\G( X ).
	\label{eq:defnOfS}
	\end{equation}
	By the convexity of $S$, $s \in S$. Given $h\in G$, we now compute that
	\begin{equation*}
	s = \E_{\G}(X)
	= \E_{\G}(X \circ \phi_{h})
	= \E_{\G}(\rho( h g) s^\prime)
	= \E_{\G}(\rho( h) \rho(g) s^\prime)
	= \rho( h) \E_{\G}( \rho(g) s^\prime)
	= \rho( h) s,
	\end{equation*}
	using \eqref{eq:rvinvariance}, \eqref{eq:defnOfX}, that $\rho$ is a
	homomorphism, linearity of expectation, and finally \eqref{eq:defnOfX}
	and \eqref{eq:defnOfS}.
	So $s$ is invariant under $G$.
	
	If $G$ is finite, the expectation is a finite sum, so we do not need the dominated convergence theorem.
\end{proof}

\begin{proof}[Proof of Lemma \ref{lemma:regularity}]
	We recall that a {\em perfect} probability measure
	is a complete probability measure, $\mu$ on a set $S$ such that for every measurable map $f:S \to \R$
	the image measure is a regular measure on $\R$. Lemma 2.4.3.\ of \cite{itoIntroduction} proves
	that all standard probability spaces are perfect. Let $S$ be a perfect probability space and
	let $V$ be a finite-dimensional real vector space, then Exercise
	3.1(iii) of \cite{itoIntroduction} shows that any measurable map $f:S \to V$ induces
	a regular measure on $V$. Thus it suffices to show that $\pi$ defined by \eqref{eq:definitionOfPi} is
	measurable.
	
	Choose a basis $\{X_i\}$ for $\dom c$. Define a map $X:\Omega \to \R^n$ by requiring that the $i$-th component of $X(\omega)$
	is given by $X(\omega)_i = X_i(\omega)$. This map is measurable since each $X_i$ is measurable. Define a map $X^{**}:(\dom c)^* \to \R^n$
	by requiring that the $i$-th component of $X^{**}(f)$ is given by $X^{**}(f)_i = f(X_i)$. $(X^{**})^{-1}$ is a linear isomorphism and so
	is measurable by the definition of the topology on $\dom c$. Since $\pi=(X^{**})^{-1}\circ X$, $\pi$ is measurable.
\end{proof}

\begin{proof}[Proof of Theorem \ref{thm:linearmarkets}]
	We first show that $\Vec(M)$ lies in $\VecM$.	
	
	We have already seen in Lemma \ref{lemma:regularity} that $d_M$ is regular.	
	
	We must also show that $d_M$ is non-degenerate. Given $X \in \dom c$
	we may define a linear functional $X^{**} \in (\dom c)^{**}$ by $X^{**}(f)=f(X)$. Double duality is an isomorphism, so given distinct $\tilde{X}, \tilde{Y} \in (\dom c)^{**}$ we may find distinct $X, Y \in (\dom c)$ with $X^{**}=\tilde{X}$ and $Y^{**}=\tilde{Y}$.
	For any $Z \in (\dom c)$, $Z^{**}\circ \pi=Z$. $M$ is separated, so $\pi$ is a mod 0 isomorphism. Since $X$ and $Y$ are not equal, it then follows that $\tilde{X}=X^{**}$ and $\tilde{Y}=Y^{**}$ are not equal almost everywhere. So $d_M$ is non-degenerate, as claimed.
	
	This completes the proof that $\Vec(M)$ lies in $\VecM$.
	
	We now define an additional map, also denoted $\Vec$, which sends morphisms of $\FinM$ to morphisms of $\VecM$. Given a market morphism $T$ between two such markets $M_i=((\Omega_i, {\cal F}_i,\P_i),c_i)\in \FinM$ ($i=1,2$) we define $T^*:\dom c_2 \to \dom c_1$ by
	$
	T^*(f)=f\circ T.
	$
	We define $\Vec(T)=T^{**}:(\dom c_1)^* \to (\dom c_2)^*$ to be the ordinary vector space dual of $T^*$.
	We wish to show that $\Vec(T)$ is a morphism in $\VecM$. Since
	$T$ is a market morphism we compute that for any $v \in (\dom c_1)^*$
	\[
	(\Vec(T) c_1) (v) = T^{**}(c_1)(v) = c_1 (T^* v) = c_1( v \circ T ) \leq c_2(v).
	\]
	Applying the same calculation to $-v$ and using linearity, we also have
	$
	(\Vec(T) c_1) (v) \leq -c_2(v).
	$
	Hence
	\begin{equation}
	(\Vec(T) c_1) (v) = c_2(v).
	\label{eq:morphismCondition1}
	\end{equation}
	We note that
	\begin{equation*}
	T^{**}(v) = w \iff \forall f \in V_2^*, \, T^*f(v) = f(w)
	\iff \forall f \in V_2^*, \, f T(v) = f(w) 
	\iff T(v) = w.
	\end{equation*}
	It follows that given a set $A \subseteq V_2$
	\[ 
	(T^{**})^{-1}(A) = T^{-1}(A).
	\]
	So if $A$ is Borel measurable we have
	\begin{equation}
	d_1(\Vec(T)^{-1} A) = d_1((T^{**})^{-1}(A)) = d_1(T^{-1}(A))= d_2(A).
	\label{eq:morphismCondition2}
	\end{equation}
	Together \eqref{eq:morphismCondition1}
	and \eqref{eq:morphismCondition2} show that $\Vec(T)$ is a morphism in $\VecM$, as claimed.
	
	We must show that $\Fin((V,d,c))$ is
	an element of $\FinM$.
	We first note that the probability space underlying $\Fin((V,d,c))$ is standard,
	since a regular distribution on a real vector space always defines a standard probability distribution.
	Since all elements of $\VecM$ have non-degenerate distributions, $\dom \underline{c} \subset L^0(V;\R)$ is equal to $V^*$ (rather than a non-trivial quotient space of $V^*$ by equivalence almost everywhere). The dual space of a finite-dimensional vector space separates the points of the vector space, so $\Fin((V,d,c))$ is separated. It is now clear that $\Fin((V,d,c))$ lies in $\FinM$.
	
	We define a mapping on morphisms, also called $\Fin$, by $\Fin(T)=T$ for any morphism $T$ of $\VecM$. We must show that a $\VecM$ morphism
	is automatically a market morphism. Equation \eqref{eq:distributionConditionForVecM} 
	shows that $\Fin(T)$ is a $\Prob$ morphism.
	
	Next observe that a
	$\VecM$ morphism is automatically surjective. Suppose for contradiction
	that $T$ is not surjective, then we can find a non-zero linear functional $X$ which annihilates $\Image(T)$.
	Since $\Image(T)$ is of full measure, $X$ is almost-surely zero, and hence $d_2$ is degenerate, yielding the desired
	contradiction.
	
	Now let $T:V_1 \to V_2$ be a morphism
	in $\VecM$ and $X \in L^0(V_2; \R)$.
	First suppose $X$ is linear, then equation \eqref{eq:costConditionForVecM} shows that $\underline{c}_1(X\circ T)=\underline{c}_2(X)$. Next suppose $X$ is not linear, so we may find $v,w \in V_2$
	and $\alpha \in \R$ with $X(\alpha v + w) \neq \alpha X(v) + X(w)$. Since 
	$T$ is surjective we may find $v^\prime, w^\prime \in V_1$ with $T v^\prime = v$
	and $T w^\prime=w$. Then $XT(\alpha v^\prime + w^\prime) \neq \alpha XT(v^\prime) + XT(w^\prime)$. So $XT$ is also non-linear, and hence $\underline{c}_1(X \circ T) = \infty = c_2(X)$. Thus $c_1(X\circ T)=c_2(X)$ for all $X \in L^0(V_2; \R)$. So
	$\Fin(T)$ is a market morphism as claimed.
	
	Since $\dom \underline{c}=V^*$, we have $(\dom \underline{c})^*=V^{**}$. Hence the composition $\Vec \circ\Fin$ is given by double duality of vector spaces. In particular $\Vec \circ \Fin(V,d,c)$ is naturally isomorphic to $(V,d,c)$.
	
	We note that $\Fin \circ \Vec (M)$ is naturally isomorphic to $M$ with the isomorphism given by $\pi$ defined in \eqref{eq:definitionOfPi}.
	
	We have now shown that $\Vec$ and $\Fin$ define an equivalence of the categories $\FinM$ and $\VecM$. It is trivial to check that vector-space duality defines
	a duality of the categories $\VecM$ and $\DualM$. The statement that $\Vec$
	and $\Dual$ define bijections follows by elementary category theory \cite{eilenbergmaclane}.
\end{proof}

\begin{proof}[Proof of Theorem \ref{thm:markowitzClassification}]
	Let $\Cov:\dom c \times \dom c \to \R$ be given by the covariance. This is a non-degenerate symmetric bilinear form and hence defines an inner product on $\dom c$. All real inner-product spaces of dimension $n$ are isomorphic to the standard Euclidean space $\R^n$, hence we can find a second basis $\{Y_i\}$ for $\dom c$ with 
	covariance matrix $\id_n$. The distribution of these assets will still be a multivariate normal distribution, but now with covariance matrix $\id_n$. This shows that the market is Gaussian.
\end{proof}

\begin{proof}[Proof of Corollary \ref{cor:generalMutualFund}]
	It suffices to prove the result for markets of the form \eqref{eq:genericMarkowitz}. Let $\phi:\R^n \to \R^n$ be the linear transformation given by the matrix
	\[
	\phi_{ij}=\begin{cases}
	1 & \text{$i=j$ and $i, j \leq 2$}, \\ 
	-1 & \text{$i=j$ and $i, j > 2$}, \\
	0 & \text{otherwise}. 
	\end{cases}
	\]	
	$\phi$ defines an automorphism of any market of the form \eqref{eq:genericMarkowitz}.
	Any invariant investment strategy must be invariant under $\phi^*$. $\phi^*$ has the same matrix representation as $\phi$ when written with respect to the standard dual basis $\{e_i^*\}$ for $(\R^n)^*$. If $X$ is an invariant investment strategy, its components $(X)_i$ written with respect to this basis satisfy $X_i=0$ for $i > 2$.
\end{proof}

\subsection{Proofs for Section \ref{sec:completeOnePeriod}}

We review the features
of the theory of disintegration of measures we will need.

\begin{definition}
	Let $\{S_\alpha\}$ be a countable collection of subsets of a set $S$. We write $\zeta(\{S_\alpha\})$ for the collection of
	sets of the form
	\[
	\bigcap_{i=1}^\infty S^\prime_\alpha, \quad (S^\prime_\alpha = S_\alpha\text{ or }
	S^\prime_\alpha = S \setminus S_\alpha).
	\]
	These sets are disjoint and cover $S$ so they define a decomposition of $S$ called the {\em decomposition generated by $\{S_\alpha\}$}. A decomposition of a measurable set $S$ generated by a countable
	collection of measurable sets is called a {\em measurable decomposition}. Here we are using the terminology of \cite{rokhlin} p5 and p26. These decompositions are called {\em separable decompositions} in \cite{itoIntroduction}.
	We say that two measurable decompositions $\zeta$ and $\zeta^\prime$ of probability spaces $\Omega$ and $\Omega^\prime$ are {\em mod 0 isomorphic} if there is a mod 0 isomorphism of $\Omega$ mapping the elements of $\zeta$ to the elements of $\zeta^\prime$.
\end{definition}

Given a decomposition $\zeta$ of a probability space $\Omega$ we may define a  projection map, $\pi_\zeta:\Omega \to \zeta$ by sending a point $\omega$ to the element of $\zeta$ containing $\omega$. This projection map induces a measure $\mu_\zeta$ on $\zeta$.
Rokhlin refers to the resulting measurable space as the quotient space $\Omega/\zeta$ (see p4 of \cite{rokhlin}).

\begin{definition}
	\label{def:canonical}
	Let $\zeta$ be a decomposition of a standard probability
	space $\Omega$.
	Let $\mu_C$ be a set of measures defined indexed by $C \in \zeta$. We say that $\mu_C$ is {\em canonical with respect to $\zeta$} if the following hold.
	\begin{enumerate}[nosep,label=(\roman*)]
		\item $\mu_C$ is a standard probability
		space for $\mu_\zeta$-almost-all $C \in \zeta$.
		\item If $A$ is a measurable subset of $\Omega$ then:
		\begin{enumerate}[nosep,label=(\alph*)]
			\item the set
			$A \cap C$ is $\mu_C$ measurable for $\mu_\zeta$-almost-all $C$;
			\item $\mu_C(A \cap C)$ defines a $\mu_\zeta$-measurable function acting on $C \in \zeta$;
			\item the measure $A$ can be recovered by integrating over $\zeta$, i.e.
			\[
			\mu(A) = \int_{\zeta} \mu_C(A \cap C) \, \ed \mu_\zeta.
			\]
		\end{enumerate}
	\end{enumerate}
\end{definition}

This definition is simply a translation of the definition on p25 of \cite{rokhlin} into our notation.
We note that what we call a standard probability space, Rohklin calls a Lebesgue space. The equivalence of these notions is given on p20 of \cite{rokhlin}.

We may now state two theorems, both due to Rohklin.
\begin{theorem}
	\label{thm:rokhlin1}
	Let $\Omega$ be a standard probability space.
	There exists a set of measures $\mu_C$ canonical with respect to $\zeta$ if and only 
	$\zeta$ is a measurable decomposition (\cite{rokhlin} p26). Moreover, $\mu_C$ is defined
	essentially uniquely: if $\mu_C$ and $\mu_{C^\prime}$ are both canonical for $\zeta$ then $\mu_C$ is mod 0 isomorphic to $\mu_{C^\prime}$ for $\mu_\zeta$-almost-all $C$ (\cite{rokhlin} p25).
\end{theorem}
\begin{theorem}
	\label{thm:rokhlin2}
	Let $\Omega$ be a standard probability space and $\zeta$ a measurable decomposition.
	Let $m_\zeta:\zeta \to {\cal S}$ be given by mapping the measure $\mu_C$ to the element of ${\cal S}$
	corresponding to its isomorphism class. Then $m_\zeta$ is $\mu_\zeta$ measurable. Two decompositions $\zeta$ and $\zeta^\prime$ are mod 0 isomorphic if and only $\mu_\zeta$ and $\mu_{\zeta^\prime}$
	are mod 0 isomorphic via a map sending $m_\zeta$ to $m_\zeta^\prime$ (\cite{rokhlin} p40).
\end{theorem}

Finally, Theorem 3.3.1 of \cite{itoIntroduction}
tells us that if $X$ is a real random variable, and if we define $\zeta$ to be the set of sets of the form $X^{-1}(x)$ then
$\zeta$ is a measurable decomposition.
When we apply Theorem \cite{rokhlin} to the level sets of a random variable $\zeta$, the measure $\mu_{X^{-1}(x)}$ on the level set $X^{-1}(x)$ for $x \in \R$ is called the {\em conditional probability measure}, conditioned on $X=x$ (see \cite{itoIntroduction} Section 3.5). Note that in this case the projection map sending the level set $X^{-1}(x)$ to $x$ defines a mod 0 isomorphism between $\zeta$ with measure $\mu_\zeta$ and the probability measure on $\R$ induced by $X$.

\begin{proof}[Proof of Theorem \ref{thm:classificationComplete}]
	First note that $(\mu_q,m_q) \in \Measures(n)$ is manifestly an invariant of $\Omega$.
	
	Given a pair $M=(\mu,m) \in \Measures(n)$, let us see how to define $\Omega(M)$ with $(\mu_q,m_q)=M$.
	
	Let $a_0$ be the probability space $[0,1]$. For $i> 0$, let $a_i$ be a probability space consisting of a single atom. We take as probability space
	\[
	\Omega(M)=(0,\infty)^n
	\times \left(\sqcup_{i=0}^\infty a_i \right).
	\]
	This has a measure we denote by $(\mu \times  \lambda)$ induced by taking the standard construction of product measures and measures on disjoint unions and then obtaining the Lebesgue extension.
	Using our concrete realisation of ${\cal S}$, given in \eqref{eq:defCalS}, we define the components $m_{i}$ of the function $m$ for $i \in \N \cup \{ \infty \}$. Let $\pi_1:\Omega_M \to (0,\infty)^n$ denote the projection onto the $(0,\infty)^n$ component. We then obtain measurable functions $m_i \circ \pi_1$ defined on $\Omega$. Given a Lebesgue measurable subset $A$ of $\Omega$, we define a measure
	$\P_{0}(A)$ by
	\begin{align}
	\P_{0}(A) &:=  \int_\Omega \sum_{i=0}^\infty (m_i\circ \pi_1) \cdot 1_{A \cap ((0,\infty)^n\times a_i)} \, \ed (\mu \times \lambda) \nonumber \\
	&=  \int_{(0,\infty)^n}  \sum_{i=0}^\infty (m_i\circ \pi_1) \int_{a_i} 1_{A \cap ((0,\infty)^n\times a_i)} \, \ed (\mu \times \lambda|_{a_i}) \nonumber \\
	&=  \int_{(0,\infty)^n}  \sum_{i=0}^\infty m_i \, \P_{a_i}({A \cap {\pi_1}^{-1}(\omega)} \cap a_i) \, \ed \mu
	=  \int_{(0,\infty)^n}  \P_m({A \cap {\pi_1}^{-1}(\omega)}) \, \ed \mu. \label{eq:definingFulfilled}
	\end{align}
	Let $\zeta$ be the decomposition of $\Omega(M)$ given by the pre-images $\pi_1^{-1}(\omega)$ for $\omega \in (0,\infty)^n$. For $\omega \in (0,\infty)^n$, let $\mu_{\pi_1^{-1}(\omega)}$ be the measure $m(\omega)$. We observe
	that $m(\omega)$ is canonical with respect to $\zeta$.
	We explicitly check the requirements given in Definition \ref{def:canonical}. Property (i) follows since $\pi_1^{-1}(\omega)$ is always standard. Similarly property (ii) (a) follows since $A \cap \pi_1^{-1}(\omega)$ is always measurable. Property (ii) (b) follows from Fubini's theorem, as used in the derivation of equation \eqref{eq:definingFulfilled} above. Property (ii) (c) is given by \eqref{eq:definingFulfilled} itself.

	For $1\leq i \leq n$, we define measures $\P_{i,M}$ by
	\begin{equation}
	\P_{i,M}(A) = \omega_i \E_\mu( \pi_1 \cdot 1_A)
	\label{eq:defPi}
	\end{equation}
	where $\omega_i$ is the $i$th coordinate function on $(0,\infty)^n$ as before.
	This is an equivalent probability measure to $\P_0$ since $\omega_i$ is positive
	and has $\P_0$ expectation of 1.
	
	We see that $\Omega(M)$ equipped with these measures satisfies $(\mu_q,m_q)=M$.
	
	Suppose $\Omega$ is a probability space with $n$ additional equivalent measures $\P_i$.
	Let $M=(\mu_q,m_q)$.
	By Theorem \ref{thm:rokhlin2} we can find a mod 0 isomorphism, $\phi$, from $\Omega$ to $\Omega_{M}$ equipped with measure $\P_0$ which also sends $q$ to $\pi_1$ for each $i$. It follows from \eqref{eq:defPi} that $\phi$ must be a $\P_i$-isomorphism too.
\end{proof}

\begin{proof}[Proof of Theorem \ref{thm:simpleCompleteMarketN}]
	Let $S$ be a standard probability space and $\zeta$ a decomposition of $S$. Let $T$ be another standard probability space. We write $\zeta \star T$ for the decomposition of $S \times T$ given by taking the product of elements of $\zeta$ with $T$. Given a set of measures $\mu_C$ on $\zeta$ we write $\mu_C \times \mu_T$ for the product measures. It is clear that if $\mu_C$ is canonical with respect to $\zeta$ then $\mu_C$ is canonical with respect to $\zeta \star T$. Thus the conditional measures of $\frac{\ed \P_i}{\ed \P_0}$ on $\Omega \times I$ are all given by products with the standard measure on $I$. Hence $(m_\Omega \times I)_0=1$ $\mu_{\Omega \times I}$-almost-everywhere.
	
	On the other hand, taking the product
	of a $\Omega$ with $I$ does not
	affect the distribution $\mu_\Omega$. So if
	we take $\Omega^\prime$ to be the space defined in the 
	statement of Theorem \ref{thm:simpleCompleteMarketN}, we will have that the invariants
	of $\Omega^\prime \times I$ are equal to
	the invariants of $\Omega \times I$.
	The result now follows from Theorem \ref{thm:classificationComplete}.
\end{proof}

\begin{proof}[Proof of Theorem \ref{thm:simpleCompleteMarket1}]
	Pick a Lebesgue measure $\P$ on $(0,\infty)$ with $\E_{\P}(\omega_1) = 1$.
	The coordinate function $\omega_1$ on $(0,\infty)$ is just the identity.
	Define a measure $\Q$ by requiring that the Radon--Nikodym derivative is $\frac{\ed \Q}{\ed \P}=\omega_1=\id$.
	
	Let $F$
	be the distribution function of this measure and $F^{-1}:[0,1]\to (0,\infty)$ its inverse
	distribution function. We equip the interval $[0,1]$ with the Lebesgue measure $\P^\prime$ and a measure $\Q^\prime$ given by requiring that the Radon--Nikodym
	derivative $\frac{\ed \Q^\prime}{\ed \P^\prime}=F^{-1}$.
	
	If we can find a simultaneous mod 0 isomorphism between the measures $(\P,\Q)$ on $M \times I$
	and $(\P^\prime, \Q^\prime)$ on $\tilde{M} \times I$ we see that Theorem \ref{thm:simpleCompleteMarket1}
	follows from Theorem \ref{thm:simpleCompleteMarketN}. We take $\P_0=\P$ and $\P_1=\Q$ when applying Theorem \ref{thm:simpleCompleteMarketN}.
	
	We will now find the required isomorphism.
	In what follows, if  $X$ is a set with measure $\mu$ we will write $X_\mu$ to emphasize the measure on $X$.
	
	Let $0\leq p_1 \leq p_2 \leq 1$.
	
	Suppose that $p_1$ and $p_2$ are the two ends of a connected component of $\image F$ then
	$F$ is continuous between $p_1$ and $p_2$ and so $F$ defines a mod $0$ isomorphism between 
	$[F^{-1}(p_1),F^{-1}(p_2))_{\P}$ and $[p_1,p_2)_{\P^\prime}$. 
	So $(F^{-1}[p_1,p_2))_{\P} \times I$ is mod 0 isomorphic to $(p_1,p_2)_{\P^\prime}\times I$ via $F \times \id$. This isomorphism maps the random variable $\omega_1 = \id$ to $F^{-1}_X$. Hence it is also
	a mod 0 isomorphism for the measures $\Q$ and $\Q^\prime$.
	
	Suppose that $p_1$ and $p_2$ are the two ends of a connected component of $[0,1]\setminus \image F$. $(F^{-1}[p_1,p_2))_\P$ is mod 0 isomorphic to the atom
	$\{F^{-1}(p_1)\}_\P$ with mass $(p_2-p_1)$.
	So $(F^{-1}[p_1,p_2))_{\P_0} \times I$ is mod 0 isomorphic to $[p_1,p_2)_{\P^\prime}$
	which in turn is mod 0 isomorphic to $[p_1,p_2)_{\P^\prime}\times I$.
	The $\Q$-measure on the atom $\{F^{-1}(p_1)\}$ is equal to
	$F^{-1}(p_1)$, which is equal to $F^{-1}(p)$ for all $p_1 \leq p \leq p_2$.
	Hence $(F^{-1}[p_1,p_2))_{\Q_0} \times I$ is simultaneously
	mod 0 isomorphic to $[p_1,p_2)_{\Q^\prime}\times I$.
	
	We may therefore cover $[0,1)\times I$ with
	a countable set of
	disjoint intervals of the form $[p_1,p_2)\times I$
	which are simultaneously $\P$/$\Q$ mod 0 isomorphic to
	$(F^{-1}[p_1,p_2)) \times I$.	
	
	We may therefore combine these mod 0 isomorphisms on intervals to obtain
	the desired mod 0 isomorphism for the $\P$ and $\Q$ measures.
\end{proof}

\begin{proof}[Proof of Corollary \ref{cor:convexMutualFund}]
	We have the obvious inclusion $\iota: L^1_{\P_0}(\Omega) \to L^1_{\P_0}(\Omega \times I)$. Any element of $L^1_{\P_0}(\Omega \times I)$ which can be written as a function of the Radon--Nikodym derivatives $\frac{\ed \P_i}{\ed \P_0}$ must lie in the image of $\iota$. Hence
	it suffices to prove that $\iota A$ contains
	an element which can be written as a function
	of these Radon--Nikodym derivatives.

	By Theorem \ref{thm:simpleCompleteMarketN} we
	may assume without loss of generality that the market $\Omega \times I$ is given by  $\Omega^\prime \times I=(0,\infty)^n \times I$ and $\P^\prime_i$ as described in Theorem \ref{thm:simpleCompleteMarketN}. In this case the Radon--Nikodym derivatives are given by the coordinate functions $\omega_i$. 
	
	Let $G = S^1 \cong \R / \Z$ with measure given by the quotient measure. Since each element of $\R / \Z$ has a unique representative on $[0,1)$, $G$ is strictly isomorphic to $[0,1)$ as a probability space. Hence we may define an action of $G$ on any product space $X \times I$ by using the action on the right-hand side of the product. We can apply Theorem \ref{thm:genericMutualFundTheorem} with this choice of $G$ and taking as $\iota A$ as the convex set. The result now follows.
\end{proof}

We collect together the key properties of rearrangement in a single lemma.

\begin{lemma}
	\label{lemma:rearrangement}
	Let $m$ be a Lebesgue measure on $(0,\infty)$ satisfying condition \eqref{eq:conditionOnM}. Then
	$U_m$ is a uniformly-distributed random variable. Let $X$ be a random variable in	$X \in L^0_{\P_m}((0,\infty) \times [0,1); \R)$.
	
	The $\P_m$ distribution is left fixed by rearrangement of $X$. The $\Q_m$ distributions are increased or decreased according to whether one applies the increasing
	or decreasing rearrangement. Symbolically:
	\begin{equation}
	d^{\P_m}(X)=d^{\P_m}(R^\pm_m(X))
	\label{eq:rearrangementFixingProperty},
	\end{equation}
	\begin{equation}
	d^{\Q_m}(X) \preceq d^{\Q_m}(R^+_m(X))
	\label{eq:increasingRearrangementProperty},
	\end{equation}
	\begin{equation}
	d^{\Q_m}(X) \succeq d^{\Q_m}(R^-_m(X))
	\label{eq:decreasingRearrangementProperty}.
	\end{equation}
	
	In addition:
	\begin{equation}
	\frac{\ed \Q_m}{\ed \P_m}(\omega)<\frac{\ed \Q_m}{\ed \P_m}(\omega^\prime) \implies
	R^\pm_m(\pm X(\omega)) \leq  R^\pm_m(\pm X(\omega^\prime))
	\label{eq:increasingInNikodym1},
	\end{equation}
	\begin{equation}
	d^{\P_m}(X) \preceq d^{\P_m}(Y) \implies d^{\Q_m}(R^+_m(X)) \preceq d^{\Q_m}(R^+_m(Y))
	\label{eq:transitivity},
	\end{equation}
	\begin{equation}
	F_{\frac{\ed \Q_m}{\ed \P_m}}\text{ is continuous at }x \implies R^\pm_m(X)(x,y_1)=R^\pm(X)_m(x,y_2) \quad \forall  y_1, y_2
	\label{eq:constantOnFibres}.
	\end{equation}
\end{lemma}
\begin{proof}
	Pick $z \in (0,1)$. Since $F_m$ is an increasing function, we can find $x_0 \in (0,\infty)$ with
	$\lim_{x^\prime \to x_0-} F_m(x) \leq z \leq \lim_{x^\prime \to x_0+} F_m(x)$. Hence we can find $y_0$ with $U_m(x_0,y_0)=z$. Since $F_m$ is increasing,
	we deduce that
	\begin{align}
	\P_m( U_m(\omega) \leq z)
	&= \P_m( x(\omega)<x_0 \text{ or } (x(\omega)=x_0 \text{ and }y(\omega)\leq y_0) ) \nonumber \\
	&= \P_m( x(\omega)<x_0) + \P_m(x(\omega)=x_0)\P_m(y(\omega)\leq y_0) \nonumber \\
	&= \lim_{x \to x_0-} F_m(x) + y (\lim_{x \to x_0+} F_m(x)
	- \lim_{x \to x_0-} F_m(x)) = z. \label{eq:uniformProperty}
	\end{align}
	We deduce first that $U_m$ is measurable since its sublevel sets are measurable. We then deduce that $U_m$ is a uniform random variable as \eqref{eq:uniformProperty} is the defining property of uniform random variables.
	
	Property \eqref{eq:rearrangementFixingProperty} of rearrangement follows immediately from the fact that $U_m$ is uniform and from the definition of rearrangement.
	
	We note that for $\alpha \in (0,1)$,
	\begin{align*}
	\inf \{ x \in \R \mid F_X(x) \geq \alpha \} \leq k 
	&\implies
	F_X(k) \geq \alpha.
	\end{align*}
	So from the definition of rearrangement
	\begin{align*}
	\P(R^+_m(X)(\omega) \leq k)
	&=\P(F^{-1}_X(U_m(\omega)) \leq k)
	=\P(\inf \{ x \in \R \mid F_X(x) \geq U_m(\omega) \} \leq k ) \\
	&\leq \P(F_X(k) \geq U_m(\omega) )
	= F_X(k).
	\end{align*}
	The last step uses \eqref{eq:uniformProperty}. We have established \eqref{eq:increasingRearrangementProperty}. Property \eqref{eq:decreasingRearrangementProperty}
	is now obvious.
	
	From the definition of $U_m$, if $x(\omega)\leq x(\omega^\prime)$ then $U_m(\omega)\leq U_m(\omega^\prime)$. $F_X$ is increasing and $x$ is equal to the Radon--Nikodym derivative $\frac{\ed \Q_m}{\ed \P_m}$. Hence \eqref{eq:increasingInNikodym1} follows.
	
	From the definition of $U_m$, $U_m(x,y)$ is independent of $y$ when $F_m$ is continuous at $x$. Hence $R^\pm_m(X)(x,y)$ is also independent of $y$.  Note that $F_m=F_\frac{\ed \Q_m}{\ed \P_m}$. This establishes \eqref{eq:constantOnFibres}.
	
	To establish \eqref{eq:transitivity} let us suppose $d^{\P_m}(X) \preceq d^{\P_m}(Y)$. This means that
	\[
	F_X(k) \geq F_Y(k) \quad \forall k \in \R
	\]
	where $F_X$ and $F_Y$ are the $\P_m$-measure distribution functions of $X$ and $Y$. 
	Hence
	\begin{equation}
	F_X^{-1}(p) \leq F_Y^{-1}(p) \quad \forall p \in [0,1].
	\label{eq:finversecomp}
	\end{equation}
	We then find
	\begin{align*}
	\Q(R^+_m(X) \leq k) = \E_m( x 1_{(R^+_m(X)\leq k)} )
	= \E_m( x 1_{(F^{-1}_X \circ U_m \leq k)} )
	&\geq \E_m( x 1_{(F^{-1}_Y \circ U_m \leq k)} ) \quad \text{by }\eqref{eq:finversecomp} \\
	&= \Q(R^+_m(Y) \leq k).
	\end{align*}
	So $d^{\Q_m}(R^+_m(X)) \preceq d^{\Q_m}(R^+_m(Y))$ as claimed.
\end{proof}

\begin{lemma}
	If $(\Omega, {\cal F}, \P)$ is a probability space, $X$ and $Y$
	are real random variables and $Z$ is an $\R^k$ random variable
	satisfying
	\[
	d^\P(X \mid Z) \preceq d^\P( Y \mid Z)
	\]	
	then $d^\P(X) \preceq d^\P( Y )$.
	\label{lemma:conditionalExpectation}
\end{lemma}
\begin{proof}
$\P(X \leq k)=\int_{\R^k} \P(X \leq k \mid Z ) \, \ed Z
\leq\int_{\R^k} \P(Y \leq k \mid Z ) \, \ed Z = \P(Y \leq k).$
\end{proof}

\begin{proof}[Proof of Theorem \ref{thm:rearrangement}]
	By Theorem \ref{thm:simpleCompleteMarketN}, we only need consider the case when
	$\Omega=(0,\infty)^n$ equipped with a measure $\mu$ satisfying
	$
	\E_\mu(x_i)=1
	$
	for each coordinate function $x_i$.
	
	Given an integer $j$, $1 \leq j \leq n$, we
	define a random $n-1$ vector $\hat{q}_j(\omega)$ consisting of all the components of $q$ except the $j$th. We write $\hat{\mu}_j$ for the measure induced on $(0,\infty)^{n-1}$ by $q_{\hat{j}}$. We write $q_j$ for the $j$th component of $q$, and write $\mu_j$ for the measure on $(0,\infty)$ induced by $q_j$.
	
	Given a random variable $X$ on $(0,\infty)^n \times [0,1)$ and a value $Q \in (0,\infty)^{n-1}$
	we may define $X_{j,Q}:(0,\infty)\times[0,1) \to \R$ by
	\[
	X_{j,Q}(x,y)=X(Q \oplus_j x, y),
	\]
	where $Q \oplus_j x$ is the vector obtained by inserting a new component with value $x$ at the $j$th index of the vector $Q$. $X_{j,Q}$ is $\hat{\mu}_j$-almost-surely measurable.
	
	Let $y$ denote the final coordinate function on $(0,\infty)^n \times [0,1)$. We define {\em conditional rearrangements} $R^+_j$ and $R^-_j$ as follows
	\[
	R^{\pm}_j(X)(\omega) := R^\pm_{\mu_j}(X_{j,\hat{q}_j(\omega)}) \left( q_j(\omega), y(\omega) \right).
	\]
	We define $R_j=R^+_j$ if $\sign j=1$, and $R_j=R^-_j$ otherwise.
	Since $X_{j,Q}$ is $\hat{\mu}_j$-almost-surely measurable, $R^{\pm}_j$ is well-defined mod $0$.
	
	We need to check that
	$R^{\pm}_j$ is measurable. We note that
	\begin{equation*}
	F^{-1}_{X_{j,\hat{q}_j(\omega)}}(p)
	= \inf \{ z \in \R \mid F_{X_{j,\hat{q}_j(\omega)}}(z) \geq p \} 
	= \inf \{ z \in \Q \mid F_{X_{j,\hat{q}_j(\omega)}}(z) \geq p \}
	\end{equation*}
	using the monotonicity of distribution functions. Define
	\[
	f(z,\omega,p)=\begin{cases}
	z & F_{X_{j,\hat{q}_j(\omega)}}(z) \geq p, \\
	\infty & \text{otherwise}.
	\end{cases}
	\]
	It is obvious from chasing through the definitions that $f$ is measurable.
	The infimum of a countable sequence of measurable functions is measurable. Hence
	$F^{-1}_{X_{j,\hat{q}_j(\omega)}}(p)$ is measurable as a function of the pair $(\omega,p)$.
	By definition
	\[
	R^+_{\mu_j}(X_{j,\hat{q}_j(\omega)})(x,y)=F^{-1}_{X_{j,\hat{q}_j(\omega)}}(U_{\mu_j}(x,y)),
	\]
	so this quantity is measurable as a function of $(\omega,x,y)$. The measurability of $R^{\pm}_j(X)$ is now immediate.
	
	We inductively define $R^*_0(X)=X$ and $R^*_j(X)=R_j(R^*_{j-1}(X))$ for $1\leq j \leq n$. We define $R(X)=R^*_n(X)$.
	
	Let us suppose as induction hypothesis that we have established for some $j<n$ that
	\begin{equation}
	\begin{split}
	d^{\P_i}(X) &= d^{\P_i}(R_j^*(X)) \quad \text{if } i=0 \text{ or } i > j,  \\
	d^{\P_i}((\sign j)X) &\preceq d^{\P_i}(R_j^*((\sign j) X)) \quad \text{otherwise}.
	\end{split}
	\label{eq:inductionHypothesis}
	\end{equation}
	We may then apply equations \eqref{eq:rearrangementFixingProperty}, \eqref{eq:increasingRearrangementProperty}, \eqref{eq:decreasingRearrangementProperty} and \eqref{eq:transitivity}
	to find
	\begin{equation}
	\begin{split}
	d^{\P_i}(X \mid \hat{q}_{j+1}) &= d^{\P_i}(R^*_{j+1}(X)) \mid \hat{q}_{j+1}) \quad \text{if } i=0 \text{ or } i > j+1,  \\
	d^{\P_i}(R^*_{j+1}((\sign j) X) \mid \hat{q}_{j+1}) &\preceq d^{\P_i}( R_{j+1}^*((\sign j) X) \mid \hat{q}_{j+1}) \text{ otherwise}.
	\label{eq:inductionDeduction}
	\end{split}
	\end{equation}
	Applying Lemma \ref{lemma:conditionalExpectation} below, we may deduce from equations \eqref{eq:inductionDeduction} that our induction hypothesis \eqref{eq:inductionHypothesis} will also hold when $j\to j+1$.
	We deduce that \eqref{eq:inductionHypothesis} holds for $0 \leq j \leq n$. This establishes properties (i) and (ii) of $R(X)$.

	For each $i$ ($0\leq i \leq n$), define a partial order $\preceq_i$ on $\R^n$
	by
	\[
	x \preceq_i y \iff \begin{cases}
	(\sign j)x_j \leq (\sign j)y_j  & 1 \leq j \leq i \\
	x_j = y_j  & i < j \leq n.
	\end{cases}
	\]
	We suppose as induction hypothesis that for some $1\leq i\leq n-1$,
	\begin{equation}
	R_{i-1}^*(X)(\omega) \leq R_{i-1}^*(X)(\omega^\prime) \quad \text{if} \quad q(\omega) \prec_i q(\omega^\prime).
	\label{eq:inductionHypothesis2}
	\end{equation}
	Write $q^a(\omega)$ for the vector containing the first $(i-1)$ components
	of $q(\omega)$, $q^b(\omega)$ for the $i$th component of $q(\omega)$ and $q^c(\omega)$ for the remaining components. So $q(\omega)=q^a(\omega)\oplus q^b(\omega) \oplus q^c(\omega)$.
	
	Suppose that $q(\omega) \prec_{i+1} q(\omega^\prime)$ then $q^a(\omega)\preceq q^a(\omega^\prime)$, $q^b(\omega) \leq q^b(\omega^\prime)$,
	$q^c(\omega) = q^c(\omega^\prime)$. We also have either: (a) $q^a(\omega)\prec q^a(\omega^\prime)$ and
		$q^b(\omega)= q^b(\omega^\prime)$;
	(b) $q^a(\omega)= q^a(\omega^\prime)$ and $q^b(\omega)< q^b(\omega^\prime)$;
	or (c) $q^a(\omega)\prec q^a(\omega^\prime)$ and $q^b(\omega)< q^b(\omega^\prime)$.

	In case (a), our induction hypothesis \eqref{eq:inductionHypothesis2} tells us that 
	\[
	R_{i-1}^*(X)(\omega) \leq R_{i-1}^*(X)(\omega^\prime). 
	\]
	Hence by property \eqref{eq:transitivity} of rearrangement
	\[
	R_i^*(X)(\omega)=R_i(R_{i-1}^*(X))(\omega) \leq R_i(R_{i-1}^*(X))(\omega^\prime) 
	= R_i^*(X)(\omega^\prime).
	\]
	
	In case (b), we may apply \eqref{eq:increasingInNikodym1} to
	the rearrangement $R_i$ of the random variable $R_{i-1}^*(X)$ to find
	that $R_i^*(X)(\omega)\leq R_i^*(X)(\omega^\prime)$. In case (c) we apply our results for case (a) and case (b) in succession and use the transitivity of $\leq$ to again find that $R_i^*(X)(\omega)\leq R_i^*(X)(\omega^\prime)$. Thus \eqref{eq:inductionHypothesis2} remains true when we change $(i-1)\to i$.
	
	The induction hypothesis \eqref{eq:inductionHypothesis2} is trivially true when $i=1$, so claim (iii) follows.
\end{proof}

\begin{proof}[Proof of Corollary \ref{cor:rearrangement}]
	Let $X \in L^0(\Omega)$. We define $\tilde{X} \in L^0(\Omega \times [0,1)$ by $\tilde{X}(\omega,y)=X(\omega)$.
	This will satisfy $d^{\P_i}(X)=d^{\P_i}(\tilde{X})$ for all $i$.
	
	Consider case (b) of our claim. By property \eqref{eq:constantOnFibres} of rearrangement, $R_j$, and hence $R$, only depends upon $q$. So we may
	write $R(\tilde{X})=\hat{X}(q)$ for some $\hat{X}$. We define $\tilde{R}(X)=\hat{X}(q)$, and it will satisfy all the desired properties.
	
	Now consider case (a) of our claim. Let us write $\{x_n\}$
	for the countable set of discontinuities of $F_{q_1}$.
	We define a set $
	\Delta_n := \left(q_1\right)^{-1} (x_n)$.
	Since the probability space is standard and atomless, there is a mod $0$
	isomorphism $\phi_n$ from the set $\Delta_n$
	to the set
	$
	\{ x_n \} \times I.
	$
	We write $\Delta=\bigcup \Delta_n$. Property \eqref{eq:constantOnFibres}
	tells us that the rearrangement $R(\tilde{X})(\omega,y)$ only depends upon
	$y$ if $x \in \Omega \setminus \Delta$. So we may define a function $\hat{X}$
	on $(0,\infty) \setminus \{ x_n \}$ by $\hat{X}(q_1)=R(\tilde{X})$ on $\Omega \setminus \Delta$. We now define
	\[
	\tilde{R}(X)(\omega)=\begin{cases}
	\hat{X}(q_1(\omega)) & \omega \in \Omega \setminus \Delta, \\
	R(\tilde{X})(\phi(X)) & \text{otherwise}.
	\end{cases}
	\]
	Since each $\phi_n$ is a mod 0 isomorphism on $\Delta_n$ and preserves the
	Radon--Nikodym derivatives, we see that
	\[
	d^{\P_i}(\tilde{R}(X))=d^{\P_i}(R(\tilde{X}))
	\]
	for $i=0,1$.  The result follows.
\end{proof}

\subsection{Proofs for Section \ref{sec:ctstime}}

Let us briefly review how the measure $\Q$ is constructed. Suppose that further to the assumptions of Definition \ref{def:exchangeMarket},
we may define a process $Z_t$ by
\begin{equation}
Z_t = \int_0^t (\bm{\sigma}^{-1}(r \bm{X}_s -  \bm{\mu})) \cdot \ed \bm{W}_s
\label{eq:defOfZT}
\end{equation}
where $\cdot$ denotes the usual inner product
of vectors. We have suppressed the parameters $(\bm{X}_s,s)$ of the functions $\bm{\sigma}$ and $\bm{\mu}$ to keep our expressions readable, and will do
this throughout this section.
We then define $q_t$ to be the Dol\'eans-Dade exponential of $Z_t$,
\begin{equation}
q_t = \exp\left( Z_t - \frac{1}{2} [Z,Z]_t \right),
\label{eq:qFirstDef}
\end{equation}
so that $q$ is a positive process and a local $\P$-martingale. If $q_t$ is a $\P$-martingale,
then the measure $\Q$ can be defined by
\begin{equation}
\Q(A) = \E_\P(q_T A)
\label{eq:defofQMeasure}
\end{equation}
for a measurable set $A \subset \Omega$.

\begin{proof}[Proof of Theorem \ref{thm:marketPriceOfRiskInvariant}]
	Applying It\^o's Lemma to the defining equation for the Dol\'eans-Dade exponential we compute that
	\[
	\int_0^t \frac{1}{q^2_s} \, \ed [q, q]_s = [Z,Z]_t.
	\]
	Hence by \eqref{eq:defOfZT}	
	\begin{equation}
	\int_0^t \frac{1}{q^2_s} \, \ed [q, q]_s = \int_0^t |\bm{\sigma}^{-1}(r \bm{X}_s -  \bm{\mu})|^2 \ed s.
	\label{eq:amprInvariant}
	\end{equation}
	Since
	\[
	q_t = \frac{\ed \Q}{\ed \P}\Big|_{{\cal F}_t},
	\]
	$q_t$ is manifestly an invariantly-defined stochastic process (for the obvious choice of functor). Hence the left-hand side of equation \eqref{eq:amprInvariant} is manifestly an invariantly-defined stochastic process. We can
	characterise the process $|\bm{\sigma}^{-1}(r \bm{X}_t -  \bm{\mu})|$ as the unique non-negative element in $A_t \in L^0(\Omega \times [0,T], \P \times \lambda)$
	satisfying
	\[
	\int_0^t \frac{1}{q^2_s} \, \ed [q, q]_s = \int_0^t A^2_s \, \ed s.
	\]
	$A_t$ defined in this way is manifestly invariantly defined, so the absolute market
	price of risk is also invariantly defined.
\end{proof}

\begin{proof}[Proof of Theorem \ref{thm:invarianceDimension}]
	Suppose for a contradiction that $n$-dimensional Wiener space, $\Omega_n$, is isomorphic to $m$-dimensional Wiener space with $m>n$. Using this isomorphism we may find $m$ independent standard Brownian motions on $\Omega_n$, $\tilde{W}^j_t$ ($1 \leq j \leq m$). By the martingale
	representation theorem, there are unique, predictable processes $\alpha^{ij}_t$ ($1 \leq i \leq n$, $1 \leq i \leq m$)
	such that
	\[
	\tilde{W}^j_t = \int_0^t \sum_{a=1}^n \alpha^{aj}_s \ed W^a_s.
	\]
	Let $\alpha_t$ be the $n \times m$ matrix with components $\alpha^{ij}$ and let 
	$\id_m$ denote the identity matrix of dimension $m$. We compute the quadratic-covariation matrix of each side in the above expression to obtain
	$\id_m = (\alpha_t) (\alpha_t)^\top$.
	Since $\alpha_t$ has rank less than or equal to $n$, and $\id_m$
	has rank $m$ we obtain the desired contradiction.
\end{proof}

\begin{proof}[Proof of Theorem \ref{thm:testcase}]
	Given such a complete market, let $Q_t$ be defined as in \eqref{eq:qtDefGeneral} and let
	\begin{equation}
	\tilde{Z}_t = \log Q_t + \frac{1}{2} \int_0^t A(s)^2 \ed s.
	\label{eq:defztilde}
	\end{equation}
	We compute
	\begin{align*}
	\ed \tilde{Z}_t &= \ed(\log Q_t) + \frac{1}{2} A(t)^2 \, \ed t
	= \frac{1}{Q_t} \, \ed Q_t - \frac{1}{2 Q^2} \, \ed [Q,Q]_t + \frac{1}{2} A(t)^2 \, \ed t
	= \frac{1}{Q_t} \, \ed Q_t.
	\end{align*}
	Hence $\tilde{Z}_t$ is a continuous local martingale. 
	We now define
	\begin{equation}
	\tilde{W}^1_t = -\int_0^t \frac{1}{A(s)} \, \ed \tilde{Z}_s.
	\label{eq:defwtilde}
	\end{equation}
	$W^1_t$ is a continuous local martingale by our assumptions on $A(t)$. We compute its
	quadratic variation.
	\begin{equation*}
	[\tilde{W}^1,\tilde{W}^1]_t =
	\int_0^t \frac{1}{A(s)^2} \ed [\tilde{Z},\tilde{Z}]_s
	= \int_0^t \frac{1}{A(s)^2} \ed [ \log Q, \log Q]_s
	= \int_0^t \frac{1}{Q_s^2 A(s)^2}\ed [Q,Q]_s = \int_0^t \ed s = t
	\end{equation*}
	by \eqref{eq:defwtilde}, \eqref{eq:defztilde}, It\^o's Lemma and 
	\eqref{eq:amprDefGeneral}.
	It follows by L\'evy's characterisation of Brownian motion that $\tilde{W}^1_t$ is Brownian motion.
	
	We may now find additional Brownian motions, $\tilde{W}^i_t$ for $2 \leq i \leq n$, such that the vector process $\bm{\tilde{W}}_t$ with components $\tilde{W}^i_t$ is a standard $n$-dimensional Brownian motion.
	
	To see this, we use the fact that $\Omega$ is assumed to be an $n$-dimensional Wiener space, so admits an $n$-dimensional standard Brownian motion $\bm{\hat{W}_t}$. Using the martingale representation theorem, we may write $\tilde{W}^1_t=\int_0^t \bm{\alpha}_s \cdot \ed \bm{\hat{W}}_s$ for a predictable vector process $\bm{\alpha}_t$ of norm 1. Given a vector $v \in \R^n$ of norm 1, we define a number $i_k$ for each $2 \leq k \leq n$ by $i_k = \inf\{i \mid \dim \langle v, e_1, e_2, \ldots, e_{i} \rangle \geq k\}$. Then $\{v, e_{i_2}, e_{i_3}, \ldots, e_{i_n}\}$ is a basis of $\R^n$. Applying the Gram--Schmidt process to this basis yields an orthonormal basis $\{v_i \}$ for $\R^n$ with $v_1=v$ and which is determined entirely by $v$. Applying this construction with $v=\bm{\alpha}_s$ we obtain a predictable orthonormal basis $\{ \bm{\alpha}^i_t \}$.
	We now define
	\[
	\tilde{W}^i_t = \int_0^t \bm{\alpha}^i_s  \cdot \ed \bm{W}_s.
	\]
	The process $\bm{\tilde{W}}_t$ is a continuous semi-martingale
	and its quadratic-covariation matrix has components
	\[
	[\tilde{W}^i, \tilde{W}^j]_t= \int_0^t \bm{\alpha}^i_s \cdot \bm{\alpha}^j_s \, \ed s = t\, \delta^{ij}.
	\]
	Hence by L\'evy's characterisation this is indeed $n$-dimensional Brownian motion.
	
	We now define a stochastic process ${\bm X}_t$ by
	\begin{equation}
	\ed {\bm{X}}_t = (r \bm{X}_t + A(t) e_1) \ed t + \ed \bm{\tilde{W}}_t.
	\label{eq:dyanmicsForX}
	\end{equation}
	Here we use the boundedness and measurability of $A$ to ensure existence
	and uniqueness of the solution to this SDE.
	The continuous-time market associated to \eqref{eq:dyanmicsForX} has $Z_t$ given by formula
	\eqref{eq:defOfZT}, so
	\begin{equation}
	\ed Z_t = - A(t)\, \ed \tilde{W}^1_t.
	\label{eq:Zsde}
	\end{equation}
	In particular $
	\ed [Z,Z]_t=A(t)^2 \ed t$, so equation \eqref{eq:qFirstDef} becomes
	\[
	\log( q_t ) = Z_T - \frac{1}{2} \int_0^t A(s)^2 \, \ed s.
	\]
	So we find
	\begin{equation*}
	\ed (\log q_t) = \ed Z_T - \frac{1}{2} A(t)^2 \, \ed t
	= -A(t) \ed \tilde{W}^1_t - \frac{1}{2} A(t)^2 \, \ed t,  \quad \text{ by \eqref{eq:Zsde}}.
	\end{equation*}
	On the other hand we compute from \eqref{eq:defztilde} and \eqref{eq:defwtilde} that
	\begin{equation*}
	\ed (\log Q_t) = \ed \tilde{Z}_t - \frac{1}{2}A(t)^2\, \ed t 
	= -A(t) \ed \tilde{W}^1_t - \frac{1}{2} A(t)^2\, \ed t.
	\end{equation*}
	Since we also have $q_0=Q_0=1$, we see that $Q_t=q_t$.
	
	Prices in $M$ are, by definition, given by
	\[
	c_t(X) = \E( e^{-r(T-t)} Q X \mid {\cal F}_t ) = \E( e^{-r(T-t)} Q_t X_t ).
	\]
	Prices in the complete market associated with \eqref{eq:dyanmicsForX}
	are given by the same formulae with $Q$ replaced by $q$. Hence the costs are the
	same in both markets, showing that we have identified a market isomorphism.
\end{proof}
\begin{proof}[Proof of Theorem \ref{thm:ctsTimeMutualFundTheorem}]
	Without loss of generality our market is a canonical Bachelier market.
	Let $A$ be an invariant convex set of martingales. Let $Y$ be an element of $A$. By the martingale 
	representation theorem
	\[
	Y_t =Y_0 + \sum_{i=1}^n \int_0^t a^i_s \, \ed W^i_s
	\]
	for some predictable processes $a^i_s$. By invariance of $A$, we see that
	\[
	Y_t =Y_0 + \int_0^t a^1_s \, \ed W^1_s  - \sum_{i=2}^n \int_0^t a^i_s \, \ed W^i_s
	\]	
	is also in $A$, as flipping the signs of the Brownian motions $W^k_t$ for
	$2\leq k\leq n$ induces an isomorphism of the canonical Bachelier model.
	
	By the convexity of $A$,
	\[
	Y_t = Y_0 + \int_0^t a^1_s \, \ed W^i_s
	\]	
	lies in $A$. Hence by the theory of \cite{harrisonPliska}, the martingale $Y_t$ can be replicated using a predictable
	self-financing trading strategy using only the asset $W^1_t$ and the
	risk-free asset. A second application of the martingale representation
	theorem shows that the asset $W^1_t$ may itself be replicated by
	a trading strategy using only the assets $X^i_t$. The hedging portfolio
	obtained in this way gives rise to the portfolio referred to in the
	statement of the theorem.
	
	We wish to compute this portfolio explicitly in the case
	of markets of the form \eqref{eq:nDDiffusion}.
	
	We may read off from \eqref{eq:defOfZT} and \eqref{eq:defwtilde} that
	\[
	\ed \tilde{W}^1_t = - \frac{1}{A(s)} \bm{\sigma}^{-1}(r\bm{X}_s-\bm{\mu}) \cdot \ed \bm{W}_s.
	\]
	From \eqref{eq:nDDiffusion} we may write
	\begin{equation*}
	\ed \tilde{W}^1_t = -\tfrac{1}{A(t)} \bm{\sigma}^{-1}(r\bm{X}_t-\bm{\mu}) \cdot ( \bm{\sigma}^{-1} (\ed \bm{X}_t - \bm{\mu} \,\ed t) )
	= -\tfrac{1}{A(t)} (\bm{\sigma} \bm{\sigma}^\top)^{-1}(r\bm{X}_t-\bm{\mu}) \cdot (\ed \bm{X}_t - \bm{\mu} \,\ed t).
	\end{equation*}
	We can now read off that the portfolio of risky assets one should
	hold in order to replicate $W^1_t$ is proportional to 
	$
	(\bm{\sigma} \bm{\sigma}^\top)^{-1}(r\bm{X}_s-\bm{\mu}).
	$
\end{proof}

\section{Basic concepts of category theory}
\label{sec:category}

In this section, we review the concepts from category theory required for this paper.

\begin{definition}
	A {\em category} $C$ consists of the following data:
	\begin{enumerate}[label=(\roman*)]
		\item a class $\ob(C)$ of {\em objects}.
		\item a class $\hom(C)$ of {\em morphisms}. To each morphism $f$
		are associated a source $a \in \ob(C)$ and target $b \in \ob(C)$. We write $f:a \to b$. $\hom(a,b)$ is the class of all morphisms from $a$ to $b$.
		\item for all $a, b, c \in \ob{C}$ a binary operation $\hom(a,b) \times \hom(b,c) \to \hom(a,c)$ called composition. If $f:a \to b$, $g:b \to c$ we write $g \circ f$ or just $g f$ for the composition.		  
	\end{enumerate}
	The composition satisfies
	\begin{enumerate}[label=(\roman*)]
		\item Associativity: If  $f:a\to b$, $g:b\to c$, $h:c\to d$
		\[	f \circ (g \circ h) = (f \circ g) \circ h \]
		\item Identity: For all $x \in \ob(C)$ there exists a morphism ${\mathbf 1}_x:x \to x$
		with the property that if $f:a \to x$, ${\mathbf 1}_x \circ f=f$ and if $g:x \to a$, $g \circ {\mathbf 1}_x = g$.
	\end{enumerate}
\end{definition}

A basic example is the category $\text{Set}$ of all ``small sets''. To define this, one first chooses a large set which contains all the sets you will be interested in. A small set is then defined to be a subset of this large set. We define the morphisms between small sets to be given by functions. One has to consider small sets rather than the category of all possible sets in order to avoid Russell's paradox. In all our definitions of categories below, the objects will be restricted to those given by small sets.

With this technicalities out of the way, we can list various familiar categories: the category $\Group$ of groups with morphisms given by homomorphism; the category $\Vec$ of vector spaces with morphisms given by linear transformations; the category $\Top$ of topological
spaces with morphisms given by continuous functions.

An isomorphism is defined to be a morphism $f$ which admits a two-sided inverse. 
An automorphism is an isomorphism whose source and target coincide. 

A basic technique in proving classification theorems is to identify invariants of the objects one is studying. 
Category theory allows us to formalize this concept.

A covariant functor is a mapping between categories and their morphisms that respects composition and identities.
\begin{definition}
	A {\em covariant functor} $F$ from a category $C$ to a category $D$ is a mapping which
	\begin{enumerate}[nosep,label=(\roman*)]
		\item associates to each object $x \in \ob(C)$ an object in $F(x) \in \ob(D)$.
		\item associates to a morphism $f:x \to y$ in $\hom(C)$ a morphism $F(f):F(x)\to F(y)$ in $\hom(D)$.
	\end{enumerate}
	and which satisfies
	\begin{enumerate}[nosep,label=(\roman*)]
		\item For all $x \in \ob(C)$, $F({\mathbf 1}_x) = {\mathbf 1}_{F(x)}$
		\item If $f:a \to b$ and $g:b \to c$ then $F(g \circ f)=F(g) \circ F(f)$.
	\end{enumerate}
\end{definition}	

A contravariant functor is a mapping between categories and their morphisms that reverses composition and identities.
\begin{definition}
	A {\em contravariant functor} $F$ from a category $C$ to a category $D$ is a mapping which
	\begin{enumerate}[nosep,label=(\roman*)]
		\item associates to each object $x \in \ob(C)$ an object in $F(x) \in \ob(D)$.
		\item associates to a morphism $f:x \to y$ in $\hom(C)$ a morphism $F(f):F(y)\to F(x)$ in $\hom(D)$.
	\end{enumerate}
	and which satisfies
	\begin{enumerate}[nosep,label=(\roman*)]
		\item For all $x \in \ob(C)$, $F({\mathbf 1}_x) = {\mathbf 1}_{F(x)}$
		\item If $f:a \to b$ and $g:b \to c$ then $F(g \circ f)=F(f) \circ F(g)$.
	\end{enumerate}
\end{definition}	

We note that in all our examples the objects are sets and the morphisms are maps between these sets. But the definitions
of category theory allow other types of object and morphism. In particular to any category one can define the {\em opposite category} by reversing the direction of morphisms. This allows one to alternatively define a contravariant functor
as a covariant functor to the opposite category.

The mapping that sends a vector space $V$ to its dual and a linear transformation to its dual is an example of a contravariant functor. The mapping that sends a vector space to its double dual is an example of a covariant functor.

We may now give a formal definition of an invariant (taken from \cite{armstrongMarkowitz}).

\begin{definition}
	Let $C$ be a category and let $F$ be a covariant functor from $C$ to $\Set$. Then an {\em invariantly-defined
	element} for $F$ is a map
	\[
	\phi: \ob(C) \to \Set
	\]
	such that $\phi(c) \in F(c)$ and $\phi(f c)=F(f)\phi(c)$ for all isomorphisms $f$ (recall that in set theory the elements of sets are themselves sets which is why the codomain of $\phi$ is $\Set$ even
	though we think of the values of $\phi$ primarily as elements rather than as sets).
	
	If $F$ is a contravariant functor, an invariantly-defined element is defined in the same way except we instead require that $\phi(f c)=F(f^{-1}) \phi(c)$
	for all isomorphisms $f$.	
	
	If $F$ is a functor from category $C$ to category $D$ and if $D$ is a category whose morphisms are in fact functions, we say that $\phi$ is an invariantly-defined element for $F$ if it is an invariantly-defined element for $U\circ F$ where $U$ is the forgetful functor.
\end{definition}

For example consider the category of smooth surfaces with morphisms given by isometries. Gauss's Theorema Egregium says that the Gaussian curvature is an invariantly-defined element for the contravariant functor
$C^\infty$ which maps a surface to the set of smooth functions on that surface.

In general, if one performs a mathematical construction which does not involve arbitrary choices on invariantly-defined
input, one will obtain an invariantly-defined output. To justify this statement rigorously one needs to show how to mirror the basic constructions of mathematics using category theory. This is discussed in more detail in \cite{armstrongMarkowitz}. As
a result we say that a mathematical object is {\em manifestly invariantly defined} if it is constructed from invariantly-defined inputs without arbitrary choices. For example the square of the Gaussian curvature on a manifold is manifestly invariantly defined once one knows that the Gaussian curvature itself is invariantly defined. What makes the Theorema Egregium remarkable, is that the Gaussian curvature is not manifestly invariantly defined.

The notion of invariantly-defined elements is closely connected to the notion of invariance under the action of a group. Given a category of groups, we may write $\Aut c$ for the group of automorphisms of an object $c$. Let $D$ be a category whose morphisms are in fact functions. Given a functor $F:C\to D$ we define an action on the set $F(c)$ by
\[
f(s)=F(f)f(s)
\]
for $f \in \Aut c$ and $s \in F(c)$. It is easy to show that if $\phi$ is invariantly defined for $F$ then $\phi(c)$
is invariant under $\Aut c$ (see \cite{armstrongMarkowitz}).

\medskip

When we come to define categories for markets, we will choose the objects and morphisms to ensure that financially
interesting questions are manifestly invariantly defined. For example, the solutions sets for portfolio optimization problems in our markets will be invariantly defined. It follows that the solution sets for portfolio optimization problems will be invariant under automorphisms of the markets. For markets with large automorphism groups, this implies significant restrictions on the possible solutions of {\em any} financially interesting question in such a market.

\medskip

One further notion that we will use from category theory is the notion of an equivalence of categories. Let us give
the necessary definitions.

\begin{definition}
	A {\em natural transformation}, $\eta$, from a functor $F:C\to D$ to a functor $G:C \to D$ is a family of morphisms
	satisfying
	\begin{enumerate}[nosep,label=(\roman*)]
		\item For each $X \in \ob(C)$ we have a morphism $\eta_X:F(X)\to G(X)$ in $D$.
		\item For every morphism $f:X \to Y$ in $C$ we have
		\[
		\eta_Y \circ F(f)=F(f) \circ \eta_X.
		\]
	\end{enumerate}
	If each $\eta_X$ is an isomorphism, $\eta$ is called a {\em natural isomorphism}.
\end{definition}

\begin{definition}
	An {\em equivalence of categories} $C$ and $D$ consists of two covariant functors $F:C \to D$ and $G:D \to C$,
	a natural isomorphism $\epsilon$ from $F\circ G$ to $\id_C$ and a natural isomorphism $\eta$ from $G \circ F$ to $\id_D$.
	Here $\id_X$ denotes the identity functor acting on a category $X$.
\end{definition}

We say that two categories are equivalent if an equivalence between the categories exists. We say that two categories $C$ and $D$ are {\em in duality} if $C$ is equivalent to the opposite category of $D$.

\end{appendices}

\end{document}